\documentclass[11pt,a4paper]{amsart}

\usepackage[top=3cm,bottom=3cm,left=3cm,right=3cm,headsep=10pt,a4paper]{geometry}

\usepackage[utf8]{inputenc}
\usepackage[draft=false,colorlinks,citecolor=blue,linktocpage,breaklinks,hypertexnames=false,pdfpagelabels]{hyperref}
\usepackage{amsmath,amssymb,amsthm,dsfont}
\usepackage{physics}
\usepackage{cleveref}
\usepackage{graphicx}
\usepackage{xcolor}

\RequirePackage[framemethod=default]{mdframed}

\newmdenv[skipabove=7pt,
skipbelow=7pt,
backgroundcolor=darkblue!15,
innerleftmargin=5pt,
innerrightmargin=5pt,
innertopmargin=5pt,
leftmargin=0cm,
rightmargin=0cm,
innerbottommargin=5pt,
linewidth=1pt]{tBox}

\newmdenv[skipabove=7pt,
skipbelow=7pt,
backgroundcolor=darkkblue!15,
innerleftmargin=5pt,
innerrightmargin=5pt,
innertopmargin=5pt,
leftmargin=0cm,
rightmargin=0cm,
innerbottommargin=5pt,
linewidth=1pt]{sBox}

\newmdenv[skipabove=7pt,
skipbelow=7pt,
backgroundcolor=blue2!25,
innerleftmargin=5pt,
innerrightmargin=5pt,
innertopmargin=5pt,
leftmargin=0cm,
rightmargin=0cm,
innerbottommargin=5pt,
linewidth=1pt]{dBox}

\definecolor{darkblue}{RGB}{0,76,156}
\definecolor{darkkblue}{RGB}{0,0,153}
\definecolor{blue2}{RGB}{102,178,255}

\theoremstyle{plain}
\newtheorem{lemma}{Lemma}[]
\newtheorem{thm}[lemma]{Theorem}
\newtheorem{stp2}{Step}
\newtheorem{stp}{Step}
\newtheorem{prop}{Proposition}
\newtheorem{cor}{Corollary}
\theoremstyle{definition}
\newtheorem{definition}{Definition}
\theoremstyle{remark}
\newtheorem{remark}{Remark}
\theoremstyle{problem}
\newtheorem{problem}{Problem}

\newenvironment{theorem}{\begin{tBox}\begin{thm}}{\end{thm}\end{tBox}}
\newenvironment{defi}{\begin{dBox}\begin{definition}}{\end{definition}\end{dBox}}

\newenvironment{step2}{\begin{sBox}\begin{stp2}}{\end{stp2}\end{sBox}}

\newcommand{\identity}{\ensuremath{\mathds{1}}}
\newcommand{\hs}{\ensuremath{\mathcal H}}
\newcommand{\ent}[2]{D ( #1 || #2 )}
\newcommand{\entA}[3]{D_{#1}(#2||#3)}

\newcommand{\enteA}[3]{D_{#1}^E(#2||#3)}

\newcommand{\R}{\mathbb{R}}
\newcommand{\Z}{\mathbb{Z}}
\newcommand{\N}{\mathbb{N}}

\newcommand{\E}{\mathbb{E}}
\newcommand{\LLL}{\mathbb{L}}

\newcommand{\LL}{\mathcal{L}}
\newcommand{\HH}{\mathcal{H}}
\newcommand{\BB}{\mathcal{B}}

\newcommand{\A}{\mathcal{A}}
\newcommand{\SSS}{\mathcal{S}}
\newcommand{\ds}{\displaystyle}
\def\ov{\overset}
\def\un{\underset}

\title{Quantum conditional relative entropy and quasi-factorization of the relative entropy}

\author[Capel]{Ángela Capel}
\author[Lucia]{Angelo Lucia}
\author[Pérez-García]{David Pérez-García}

\address[Capel]{Instituto de Ciencias Matemáticas (CSIC-UAM-UC3M-UCM), C/ Nicolás Cabrera 13-15, Campus de Cantoblanco, 28049 Madrid, Spain}
\email{angela.capel@icmat.es}

\address[Lucia]{QMATH, Department of Mathematical Sciences, University of Copenhagen, Universitetsparken 5, 2100 Copenhagen, Denmark and NBIA, Niels Bohr Institute, University of Copenhagen, Blegdamsvej 17, 2100 Copenhagen, Denmark}

\email{angelo@math.ku.dk}

\address[Pérez-García]{Departamento de Análisis Matemático, Universidad Complutense de Madrid, 28040 Madrid, Spain and 
Instituto de Ciencias Matemáticas (CSIC-UAM-UC3M-UCM), C/ Nicolás Cabrera 13-15, Campus de Cantoblanco, 28049 Madrid, Spain}

\email{dperezga@ucm.es}

\date{\today}

\begin{document}
\maketitle

\begin{abstract}
 The existence of a positive log-Sobolev constant implies a bound on the mixing time of a quantum dissipative evolution under the Markov approximation. For classical spin systems, such constant was proven to exist, under the assumption of a mixing condition in the Gibbs measure associated to their dynamics, via a quasi-factorization of the entropy in terms of the conditional entropy in some sub-$\sigma$-algebras.

In this work we analyze analogous quasi-factorization results in the quantum case. For that, we define the \textit{quantum conditional relative entropy} and prove several \textit{quasi-factorization} results for it. As an illustration of their potential, we use one of them to obtain a positive log-Sobolev constant for the heat-bath dynamics with product fixed point.
\end{abstract}

\section{Introduction}

Quantum dissipative evolutions which are governed by local Lindbladians model several kinds of noise in quantum many body systems. Their study is, hence, fundamental for the field of theoretical and experimental quantum physics in general, as well as for the implementation of quantum memory devices in particular.

In the theoretical proposal of dissipative state engineering, made in 2009, by Verstraete et al. \cite{diss-engin} and Kraus et al. \cite{kraus}, the authors proposed that a robust way of constructing interesting quantum systems which preserve the coherence for longer periods might be based on these systems. They base this proposal in the dissipative nature of noise, since it eliminates the problem of having to initialize the system carefully, due to the fact that the system is driven to a stationary fixed state that is independent of the initial state. Moreover, some experimental results of the past few years have given value to this proposal, inducing a remarkable growth in the interest on such systems. 

Concerning these systems, there are two main topics to study. On the one hand, the \textit{mixing time}, i.e., the time that it takes for an initial state to reach the fixed point, since it indicates how fast the evolution converges to its fixed point. On the other hand, the structure of the \textit{fixed point} (or set of fixed points) of the evolution, as it provides long-term properties to the system. 

Several bounds for the mixing time can be obtained by means of the optimal constants for some quantum functional inequalities, such as the \textit{spectral gap} for the \textit{Poincaré inequality} \cite{chi2} and the  \textit{logarithmic Sobolev constant} for the  \textit{log-Sobolev inequality} \cite{kast-temme}.  In this work we will focus on the latter, which provides an exponential improvement with respect to the spectral gap. 

In \cite{clasico} and \cite{cesi}, the authors consider a classical spin system in a finite lattice, in which the spins take values in the set of natural numbers, and show that, for a certain class of dynamics of this system, a modified log-Sobolev inequality is satisfied, under the assumption of a mixing condition in the Gibbs measure associated to this dynamics, improving the seminal result of \cite{marti-oliv}. The key step in the proof is a \textit{quasi-factorization of the classical entropy} in terms of the conditional entropy  in some sub-$\sigma$-algebras. 

The main purpose of this paper is to present quantum analogues of the aforementioned classical result of quasi-factorization of the entropy. For that, we will define the notion of \textit{quantum conditional relative entropy}.

Those results, following the steps of \cite{clasico} and \cite{cesi} in the classical case, and those of \cite{kast-brand} (where the authors obtained a positive spectral gap from a result of quasi-factorization of the variance), open a way to prove the existence of a positive log-Sobolev constant. In the last part of this work, we illustrate that and prove that this indeed holds for the heat-bath dynamics when the associated fixed point is product.

The paper is organized as follows: In Section \ref{sec:prelim}, we introduce some preliminaries which are necessary for the rest of the manuscript. More specifically, we set up notation, recall the definition and some basic properties of Schatten p-norms, and also present some properties concerning the von Neumann entropy of a state and the relative entropy of two states. After that, we introduce the concept of conditional expectation and show a specific example that will be of use in subsequent sections.  Finally, we recall the classical result of quasi-factorization of a function whose quantum analogue we will present in this work. 

In Section \ref{sec:cre}, we define the \textit{conditional relative entropy} of two states on a subregion and characterize axiomatically that definition. Later, we remove one property from  the definition of conditional relative entropy and present an example of this modified definition, which we call \textit{conditional relative entropy by expectations}, and compare both quantities. In the last part of this section, we show that both definitions extend their classical analogue.

In Section \ref{sec:qf}, we present the desired results of \textit{quasi-factorization of the relative entropy}. For that, we show in increasing order of difficulty several results of quasi-factorization, classifying them into two classes depending on the possible overlap of the subregions where the relative entropy is conditioned and the number of such subregions.  Moreover, we remark that one of them is equivalent to a result that was proven in \cite{us}.

Finally, in Section \ref{sec:logSob}, we review  the well-known result that a log-Sobolev constant provides an upper bound for the mixing time in a quantum spin lattice, and prove that for the heat-bath dynamics with product fixed point there exists a positive log-Sobolev constant.

\section{Preliminaries}\label{sec:prelim}

\subsection{Notation}

This paper concerns finite dimensional Hilbert spaces. In most of the paper, we will consider a bipartite $\HH_{AB}= \HH_A \otimes \HH_B$ or tripartite $\HH_{ABC}= \HH_A \otimes \HH_B \otimes \HH_C $   Hilbert space. In general, for $\Lambda$ a set of $\abs{\Lambda}$ parties, we will denote by $\HH_\Lambda = \un{x \in \Lambda}{\bigotimes} \HH_x$ the corresponding $\abs{\Lambda}-$partite finite dimensional Hilbert space.

For every $\hs_\Lambda$, we denote the set of bounded linear operators on $\hs_\Lambda$ by $\BB_\Lambda= \BB(\hs_\Lambda)$, and its subset of Hermitian operators by $\A_\Lambda \subseteq \BB_\Lambda$. These elements will be called \textit{observables} and denoted by lowercase Latin letters, such as $ f_\Lambda, g_\Lambda$, which we will usually write with a subscript denoting where they are defined to avoid confusion. We also denote the set of density operators by $\SSS_\Lambda= \qty{f_\Lambda \in \A_\Lambda\, : \, f_\Lambda \geq 0 \text{ and } \tr[f_{\Lambda}]=1}  $. Its elements will be also called \textit{states} and denoted by lowercase Greek letters, like $\rho_\Lambda, \sigma_\Lambda$. 

A \textit{quantum channel} \cite{wolf} is a completely positive and trace preserving map. We call a linear map $\mathcal{T}: \BB_\Lambda \rightarrow \BB_\Lambda $  a \textit{superoperator}, and say that it is \textit{positive} if it maps positive operators to positive operators. Moreover, we call $\mathcal{T}$ \textit{completely positive} if, given $ \mathcal{M}_n $ the space of complex $n \times n$ matrices, $\mathcal{T} \otimes \identity : \mathcal{B}_\Lambda \otimes \mathcal{M}_n \rightarrow \mathcal{B}_\Lambda \otimes \mathcal{M}_n $ is positive for every $n \in \N$. Finally, we say that $\mathcal{T}$ is \textit{trace preserving} if $\tr[\mathcal{T}(f_\Lambda)]= \tr[f_\Lambda]$ for all $f_\Lambda \in \BB_\Lambda$. 

In the case of a bipartite Hilbert space, when $\Lambda=AB$, there is a natural inclusion of $\A_{A}$ in $\A_{AB}$ by identifying $\A_{A} = \A_{A} \otimes \identity_{B} $. We will say that an operator $f_{AB}\in \A_{AB}$ has support on $A$ if it can be written as $f_A \otimes \identity_B$, for a certain operator $f_A \in \A_A$. Throughout the whole paper, we will make use of the modified partial trace, which we will denote, in a slight abuse of notation, by $\tr_A: f_{AB} \mapsto \identity_A \otimes f_B$, where $f_B = \tr_A[f_{AB}]$ is the usual partial trace. 

\subsection{Schatten p-norms and noncommutative $\LLL_p$ spaces}

In the following sections, we will make use of some results concerning Schatten p-norms. Let $\hs_\Lambda$ be a separable Hilbert space and $T \in \BB_\Lambda$. Given $p \in [1, \infty)$, the \textit{Schatten $p$-norm} of $T$ is given by:
\begin{center}
$\ds  \norm{T}_p:=(\tr[ \abs{T}^p] )^{1/p} $, 
\end{center}
where
\begin{center}
$ \ds \abs{T}:= \sqrt{T^* T} $, 
\end{center}
and $T^*$ is the dual of $T$ with respect to the Hilbert-Schmidt product. If $T$ is a positive semi-definite operator, we have $\norm{T}_p=(\tr[ T^p] )^{1/p} $. For every $p \in [1, \infty)$, it is a norm, and  $\norm{\cdot}_\infty:= \un{p \rightarrow \infty}{\text{lim}} \norm{\cdot}_p$ coincides with the operator norm. For $p=1$ it is indeed the trace norm. However, for $p<1$, this is no longer a norm, since it does not satisfy the triangle inequality.

In the following proposition, we collect some interesting properties that Schatten p-norms satisfy \cite{bhatia},  \cite{pisier-xu}.

\begin{prop}\label{prop:schatten}
 Let $\hs_\Lambda$ be a separable Hilbert space and $S, T \in \BB_\Lambda$. Let $p \in [1, \infty ]$, and consider the Schatten p-norm, extending the definition in $\infty$ by $\norm{\cdot}_\infty:= \un{p \rightarrow \infty}{\text{lim}} \norm{\cdot}_p$ and taking $p=\infty$ as the dual of $p=1$. The following properties hold:
\begin{enumerate}
\item \textbf{Monotonicity.} For $1\leq p \leq p' \leq \infty$, $\norm{T}_1 \geq \norm{T}_p \geq \norm{T}_{p'} \geq \norm{T}_\infty$.
\item \textbf{Duality.} For $q\in [1, \infty]$ such that $\ds \frac{1}{p} + \frac{1}{q}=1$, $\norm{S}_q = \text{sup} \qty{\abs{ \left\langle  S,T\right\rangle } \, | \, \norm{T}_p=1}  $,\\
where $ \left\langle  S,T\right\rangle = \tr[S^* T]$ is the Hilbert-Schmidt inner product.
\item \textbf{Unitary invariance.} $\norm{UTV}_p=\norm{T}_p$ for all unitaries $U, V$.
\item \textbf{Hölder's inequality.} For $q\in [1, \infty]$ such that $\ds \frac{1}{p} + \frac{1}{q}=1$, $\norm{ST}_1 \leq \norm{S}_p \norm{T}_q$.
\item \textbf{Sub-multiplicativity.} $\norm{ST}_p \leq \norm{S}_p \norm{T}_p$.
\end{enumerate}

\end{prop}

In the definition of non-commutative $\LLL_p$, one can consider a similar norm with a state $\rho$ acting as a weight, i.e., one can use a $\rho$-weighted inner product to define a non-commutative $\LLL_p$ space. Non-commutative $\LLL_p$ spaces are equipped with a \textit{weighted inner product}, which, for a full rank state $\rho\in \SSS_\Lambda$, is given by 
\begin{center}
$\ds \left\langle f, g \right\rangle_\rho := \tr[\sqrt{\rho} f \sqrt{\rho} g]  \, $ for every $\,f,g \in \A_\Lambda$. 
\end{center}
Analogously, for every $1 \leq p  \leq \infty$, the \textit{non-commutative $\LLL_p$-norm} is given by
\begin{center}
$\ds \norm{f}_{p,\rho} := \tr[\abs{\rho^{1/2p} f \rho^{1/2p}}^p]^{1/p}  \, $ for every $\,f,g \in \A_\Lambda$. 
\end{center} 

%Some other properties for Schatten p-norms \cite{pisier-xu}.

%\begin{prop}\label{prop:schatten2}

 %Let $\hs_\Lambda$ be a separable Hilbert space, $S, T \in \BB_\Lambda$ and $p \in [1, \infty ]$. Consider $0<r,q \leq \infty$ such that $\ds \frac{1}{r}=\frac{1}{p} + \frac{1}{q} $. The following properties hold:

%\begin{enumerate}
%\item $\norm{T}_p=\norm{T^*}_p$.

%\item \textbf{Minkowski's inequality.} $\norm{T+S}_p \leq \norm{T}_p + \norm{S}_p$.

%\item \textbf{General Hölder's inequality.} $\norm{TS}_r \leq \norm{T}_p \norm{S}_q$.

%\item $\norm{T}_{2p}^2= \norm{TT^*}_p$.  
%\end{enumerate}

%\end{prop}  

\subsection{Von Neumann entropy and relative entropy}

Let $\hs_\Lambda$ be a finite dimensional Hilbert space, and $\rho_\Lambda \in \SSS_\Lambda$. The \textit{von Neumann entropy}, or just \textit{quantum  entropy}, of $\rho_\Lambda$ is given by:
\begin{equation}
S ( {\rho_\Lambda} ) := - \tr \left[ {\rho_\Lambda} \log \rho_\Lambda \right].
\end{equation}

This quantity is widely used in quantum statistical mechanics and is named after John von Neumann. In the following proposition we collect some properties of the von Neumann entropy that will be of use in further sections.

\begin{prop}[Properties of the von Neumann entropy, \cite{wehrl}, \cite{strongsub}]\label{prop:vNEprop} \text{\phantom{s}}\\
Let $\hs_{AB}$ be a bipartite finite dimensional Hilbert space, $\hs_{AB} = \hs_A \otimes \hs_B$. Let $\rho_{AB} \in \SSS_{AB}$. The following properties hold:
\begin{enumerate}
\item \textbf{Continuity.} The map $\rho_{AB} \mapsto S( {\rho_{AB}}) $ is continuous.
\item \textbf{Nullity.} $S(\rho_{AB})$ is zero if, and only if, $\rho_{AB}$ represents a pure state.
\item \textbf{Maximality.} $S(\rho_{AB}) $ is maximal, and equal to $\log N$, for $N=\text{dim}(\hs_{AB})$, when $\rho_{AB}$ is a maximally mixed state.
\item \textbf{Additivity.}  $S(\rho_A \otimes \rho_B) = S(\rho_A) + S(\rho_B)$.
\item \textbf{Subadditivity.} $S(\rho_{AB}) \leq S(\rho_A) + S(\rho_B)$.
\item \textbf{Strong subadditivity.} For any three systems $A$, $B$ and $C$,
\begin{center}
$\ds S(\rho_{ABC}) + S(\rho_B) \leq S(\rho_{AB}) + S(\rho_{BC})  $.
\end{center}
\end{enumerate}

\end{prop}

We introduce now  a measure of distinguishability of two states that will be strongly used throughout the whole manuscript, and mention some of its more fundamental properties. Let $\hs_\Lambda$ be a finite dimensional Hilbert space, and $\rho_\Lambda	, \sigma_\Lambda \in \SSS_\Lambda$. The \textit{quantum relative entropy} of $\rho_\Lambda$ and $\sigma_\Lambda$ is given by:
\begin{equation}
\ent {\rho_\Lambda} {\sigma_\Lambda} := \tr \left[ {\rho_\Lambda} (\log \rho_\Lambda - \log \sigma_\Lambda) \right].
\end{equation}

We can find in the next proposition some fundamental properties of the relative entropy. 

\begin{prop}[Properties of the relative entropy, \cite{wehrl}, \cite{strongsub}]\label{prop:REprop} \text{\phantom{s}}\\
Let $\hs_{AB}$ be a bipartite finite dimensional Hilbert space, $\hs_{AB} = \hs_A \otimes \hs_B$. Let $\rho_{AB}, \sigma_{AB} \in \SSS_{AB}$. The following properties hold:
\begin{enumerate}
\item \textbf{Continuity.} The map $\rho_{AB} \mapsto \ent {\rho_{AB}} {\sigma_{AB}} $ is continuous.
\item \textbf{Non-negativity.} $\ent {\rho_{AB}} {\sigma_{AB}} \geq 0$ and $\ent {\rho_{AB}} {\sigma_{AB}}=0 \Leftrightarrow {\rho_{AB}}=\sigma_{AB}$.
\item \textbf{Finiteness.} $\ent {\rho_{AB}} {\sigma_{AB}} < \infty $ if, and only if, $\text{supp}(\rho_{AB}) \subseteq \text{supp}(\sigma_{AB})$, where supp stands for support.
\item \textbf{Monotonicity.} $\ent {\rho_{AB}} {\sigma_{AB}} \geq \ent {T(\rho_{AB})} {T(\sigma_{AB})}$ for every quantum channel $T$.
\item \textbf{Additivity.} $\ent {\rho_A \otimes \rho_B} {\sigma_A \otimes \sigma_B}= \ent {\rho_A} {\sigma_A} + \ent {\rho_B} {\sigma_B}$.
\item \textbf{Superadditivity.} $\ent {\rho_{AB}} {\sigma_A \otimes \sigma_B}\geq \ent {\rho_A} {\sigma_A} + \ent {\rho_B} {\sigma_B}$.
\end{enumerate}
\end{prop}

Using some properties of the previous two propositions, one can prove the following well-known result, which will be of use in the following sections. We include a proof for completeness.

\begin{prop}\label{prop:sigmaprod}
 Let $\hs_{ABC}=\hs_A \otimes \hs_B \otimes \hs_C $ and  $\rho_{ABC} \in \SSS_{ABC}$. Then,
\begin{center}
$\ds I_\rho  (A:BC)  \geq I_\rho (A:B) $,
\end{center}
where $I_\rho(A:B)=\ent {\rho_{AB}} {\rho_A \otimes \rho_B}$ is the mutual information. 
\end{prop}

\begin{proof}

We have
\begin{eqnarray*}
 I_\rho  (A:BC)  - I_\rho (A:B) &=& \tr \left[ \rho_{ABC}(\log \rho_{ABC}  - \log \rho_{A}\otimes \rho_{BC})\right] \\
 & &  -  \tr[\rho_{AB}(\log \rho_{AB} - \log \rho_{A}\otimes \rho_{B})] \\
 &=& \tr[\rho_{ABC}(\log \rho_{ABC} - \log \rho_{A}\otimes \rho_{BC} -\log \rho_{AB} + \log \rho_{A}\otimes \rho_{B})] \\
 &=& \tr[\rho_{ABC}(\log \rho_{ABC} - \log \rho_{BC} -\log \rho_{AB} + \log \rho_{B})] \\
 &=& - S[\rho_{ABC}] + S[\rho_{BC}] + S[\rho_{AB}] - S[\rho_B] \geq 0,
\end{eqnarray*}
where we are using the property of strong subadditivity of Proposition \ref{prop:vNEprop} in the last inequality \cite{strongsub}. We are also using the fact that the logarithm of a tensor product is the sum of logarithms (tensored with the identity) and the following property of the trace: If $A \subseteq \Lambda$, $f_A \in \A_A$ and $g_\Lambda \in \SSS_\Lambda $, then
\begin{center}
$\ds \tr[f_A g_\Lambda] = \tr[f_A g_A]  $.
\end{center}
\vspace{-0.4cm}
\end{proof}

The difference between the two terms in the statement of this proposition is called \textit{conditional mutual information}. This result may be seen, hence, as the positivity of this quantity.

\subsection{Conditional expectation}\label{subsec-cond-exp}

In this subsection, we introduce a set of maps called \textit{conditional expectations}, which we denote by $\E$ (Section 3 of \cite{kast-brand}, \cite{libropetz}).

\begin{definition}
Let $\hs_{AB} = \hs_A \otimes \hs_B$ be a bipartite Hilbert space, and $\sigma_{AB}$ a full rank state on $\hs_{AB} $. We define a \emph{conditional expectation of $\sigma_{AB}$ on $\hs_{B}$} by a map $\E_A: \hs_{AB}  \to \hs_{B}$ that satisfies the following:
\begin{enumerate}
\item \textbf{Complete positivity.} $\E_A$ is completely positive and unital. 
\item \textbf{Consistency.} For every $f_{AB} \in \A_{AB}$,
\begin{center}
  $\tr[\sigma_{AB} \E_A(f_{AB})]=\tr[\sigma_{AB} f_{AB}]$. 
\end{center}
\item \textbf{Reversibility.} For every $f_{AB},g_{AB} \in \A_{AB}$,
\begin{center}
 $\langle \E_A(f_{AB}) , g_{AB} \rangle_{\sigma_{AB}} = \langle f_{AB}, \E_A(g_{AB}) \rangle_{\sigma_{AB}} $. 
\end{center}
\item \textbf{Monotonicity.}  For every $f_{AB} \in \A_{AB}$  and $n \in \N$, 
\begin{center}
$ \langle \E_A^n(f_{AB}), f_{AB} \rangle_{\sigma_{AB}} \geq \langle
  \E_A^{n+1}(f_{AB}), f_{AB} \rangle_{\sigma_{AB}}$.
\end{center} 
\end{enumerate}
\end{definition}

\begin{remark}
From the properties in the definition of conditional expectation, we have:
\begin{itemize}
\item Property (2) yields the fact that $\E_A^*(\sigma_{AB}) =  \sigma_{AB}$, where the dual is taken with respect to the Hilbert-Schmidt  scalar product.
  \item From property (3) we can deduce that  $\E_A$ is self-adjoint in $L_2(\sigma_{AB})$.
\end{itemize}

\end{remark}

We consider now a specific example of conditional expectation. Let $\hs_{AB}= \hs_A \otimes \hs_B$ and $\sigma_{AB} \in \SSS_{AB}$ a full-rank state. We define the \textit{minimal conditional expectation} of $\rho_{AB}\in \SSS_{AB}$ with respect to $\sigma_{AB}$ on $A$ by
\begin{equation}
\E^\sigma_A(\rho_{AB}):=\tr_A[\eta_A^\sigma \, \rho_{AB} \, \eta_A^{\sigma \dagger}],
\end{equation}
where $\eta_A^\sigma:=(\tr_A[\sigma_{AB}])^{-1/2} \sigma_{AB}^{1/2}$. This map has also been previously called \textit{coarse graining map} and \textit{block spin flip}, among other names \cite{petzcorse, majewski}. Recalling that $\tr_A[\rho_{AB}]=\rho_B$, we can write
\begin{center}
$\ds \E^\sigma_A(\rho_{AB})= \sigma_B^{-1/2} \, \tr_A[\sigma^{1/2}_{AB} \, \rho_{AB} \, \sigma^{1/2}_{AB} ] \, \sigma_B^{-1/2} $.
\end{center}

If we recall now that the partial trace is tensored with the identity in $A$, we can see that $\E^\sigma_A(\rho)$ is a Hermitian operator and, indeed, $\E^\sigma_A$ is a conditional expectation with respect to $\sigma_{AB}$ (\cite[Proposition 10]{kast-brand}). Notice that $(\E^\sigma_A)^*$, the adjoint of $\E_A^\sigma$ with respect to the Hilbert-Schmidt product, which we hereafter denote by $\E^*_A$ to simplify the notation (since we are always considering conditional expectations with respect to $\sigma_{AB}$), is given by
\begin{equation}
\E^*_A(\rho_{AB}):=\sigma_{AB}^{1/2} \, \sigma_{B}^{-1/2} \, \rho _{B} \, \sigma_{B}^{-1/2} \, \sigma_{AB}^{1/2}.
\end{equation}

This map coincides with the Petz recovery map \cite{petzrecov} for the partial trace $\tr_A$ and is a quantum channel. In particular, for every density matrix $\rho_{AB} \in S_{AB}$, $\E^*_A(\rho_{AB})$ is also a density matrix.  

This is the conditional expectation we are going to consider hereafter. One should remember that the subscript is used in the same sense as in the partial trace, i.e., denoting the subsystem which is being removed, not the one which is being kept.

\subsection{Classical case}

In \cite{clasico}, the authors consider a spin system in a finite lattice, whose spins take values in the set of positive integers, and show that, for a certain class of dynamics of this system, under the assumption of a mixing condition in the Gibbs measure associated to this dynamics, a modified log-Sobolev inequality is satisfied. For that, they first need to prove a result of quasi-factorization of the entropy of a function in terms of a conditional entropy defined in sub-$\sigma$-algebras of the initial $\sigma$-algebra. 

Consider a probability space $(\Omega, \mathcal{F}, \mu)$ and define, for every $f>0$, the \textit{entropy} of $f$ by
\begin{center}
$\text{Ent}_\mu (f):= \mu(f \log f ) - \mu(f) \log \mu(f)$.  
\end{center}

Given a sub-$\sigma$-algebra $\mathcal{G}\subseteq \mathcal{F}$, we define the \textit{conditional entropy} of $f$ in $\mathcal{G}$ by
\begin{center}
$\text{Ent}_\mu (f \mid \mathcal{G}):= \mu(f \log f \mid \mathcal{G}) - \mu(f\mid \mathcal{G}) \log \mu(f\mid \mathcal{G})$,
\end{center}
where $\mu(f \mid \mathcal{G})$ is given by
\begin{center}
$\ds  \int_G \mu(f \mid \mathcal{G}) \, d \mu = \int_G f d \mu  \, \, $ for each $G \in \mathcal{G}$.
\end{center}

With these definitions, they prove the following result of quasi-factorization of the entropy.

\begin{lemma}[Quasi-factorization. Lemmas 5.1 and 5.2 of \cite{clasico}]\label{lemma:clasico}
Let $(\Omega, \mathcal{F}, \mu)$ be a probability space, and $\mathcal{F}_1, \mathcal{F}_2$ sub-$\sigma$-algebras of $\mathcal{F}$. Suppose that there exists a probability measure $\bar{\mu}$ that makes $\mathcal{F}_1$ and $ \mathcal{F}_2$ independent, $\mu \ll \bar{\mu}$ and $\mu\mid\mathcal{F}_i = \bar{\mu}\mid\mathcal{F}_i$ for $i=1,2$. Then, for every $f \geq 0$ such that $f \log f \in L^1(\mu)$ and $\mu(f)=1$,
\begin{center}
$\ds \text{Ent}_\mu(f) \leq \frac{1}{1-4\norm{h-1}_\infty} \, \mu  \left[   \text{Ent}_\mu(f \mid \mathcal{F}_1) +  \text{Ent}_\mu(f \mid \mathcal{F}_2)\right]$,
\end{center}
where $\ds h= \frac{d\mu}{d \bar{\mu}}$ is the Radon-Nikodym derivative of $\mu$ with respect to $\bar{\mu}$.
\end{lemma}

Some quantum analogues for this result will be presented in the following sections. However, first, we need to introduce the notion of quantum conditional relative entropy.

\section{Conditional relative entropy}\label{sec:cre}

In this section, we introduce an axiomatic definition of \textit{conditional relative entropy}. We aim to present a concept that, given the value of the distinguishability between two states in a certain subsystem, quantifies their distinguishability in the whole space. In other words, for two states in a bipartite Hilbert space $\hs_A \otimes \hs_B$, the conditional relative entropy in $A$ should provide the effect of the relative entropy of those states in the global space conditioned to the value of their relative entropy in $B$, extending the classical definition of \textit{conditional entropy} of a function.  

Providing axiomatic definitions or presenting axiomatic characterizations for information theory measures is a natural problem in quantum information theory. In particular, one can find in the literature several characterizations for the relative entropy, or related quantities (see \cite{axchar5}, \cite{axchar1}, \cite{axchar2}, \cite{axchar3}, \cite{axchar4}, among others). However, for our definition,  we rely on the recent work \cite{axcharRE}, where the authors present an axiomatic characterization of the relative entropy, using strongly a previous result of Matsumoto \cite{matsumoto}. Indeed, they show that the properties of continuity (with respect to the first state), monotonicity, additivity and superadditivity are sufficient for a function to be the relative entropy. They base their proof in two facts: The first one is that the properties of continuity, additivity and superadditivity imply what they call lower asymptotic semicontinuity, and the other one, the aforementioned result \cite{matsumoto}, where the author proved that any function satisfying monotonicity, additivity and lower asymptotic semicontinuity is the relative entropy itself.

We present the concept of quantum relative entropy as a function of two states verifying certain properties. The property of monotonicity is not expected to hold in such concept, since the effect of $A$ and $B$ is not considered equally in the conditional relative entropy in $A$, so an arbitrary quantum channel (more specifically, the partial trace in $B$) is not expected to decrease this quantity. For the same reason, additivity and superadditivity are neither expected to be true; however, for them, a property with the same spirit can be considered.

\begin{defi}\label{def:CRE}
Let $\hs_{AB}=\hs_A \otimes \hs_B$. We define a \textit{conditional relative entropy} in $A$ as a function
\begin{center}
$\ds D_A(\cdot || \cdot) : \SSS_{AB} \times \SSS_{AB} \rightarrow \R_0^+ $
\end{center} 
verifying the following properties  for every $\rho_{AB}, \sigma_{AB} \in \SSS_{AB}$:
\begin{enumerate}
\item  \textbf{Continuity: } The map $\rho_{AB} \mapsto  \entA A {\rho_{AB}} {\sigma_{AB}}$ is continuous.
\item \textbf{Non-negativity: } $ \entA A {\rho_{AB}} {\sigma_{AB}} \geq 0 $  and 
\begin{enumerate}
\item[(2.1)] $\entA A {\rho_{AB}} {\sigma_{AB}} $=0 if, and only if, $\rho_{AB}= \E^*_A( \rho_{AB} )$.
\end{enumerate}
\item \textbf{Semi-superadditivity: }$ \entA A {\rho_{AB}} {\sigma_{A} \otimes \sigma_B} \geq \ent {\rho_{A}} {\sigma_{A}} $ and
\begin{enumerate}
\item[(3.1)] \textbf{Semi-additivity:} if $\rho_{AB}=\rho_A \otimes \rho_B$,  $ \entA A {\rho_{A}\otimes \rho_B} {\sigma_{A} \otimes \sigma_B} = \ent {\rho_{A}} {\sigma_{A}}   $.
\end{enumerate}
\item \textbf{Semi-monotonicity:} For every quantum channel $\mathcal{T}$, the following inequality holds:
\begin{align*}
 D_A(\mathcal{T}(\rho_{AB})|| \mathcal{T}(\sigma_{AB})) +&  D_B((\tr_A \circ \mathcal{T}) (\rho_{AB}) || (\tr_A \circ \mathcal{T}) (\sigma_{AB}) ) \\
& \leq D_A(\rho_{AB}|| \sigma_{AB}) + D_B( \tr_A(\rho_{AB}) || \tr_A(\sigma_{AB}) ),
\end{align*}
where $\ds D_B({\rho_{AB}} || {\sigma_{AB}} ) $ is the conditional relative entropy in $B$.
\end{enumerate}

\end{defi}

\begin{remark}\label{remark:recoverabilityRE}
From property (3.1) we can deduce that if we consider states with support in $A$, we recover the usual definition of relative entropy, i.e.,
\begin{center}
$\ds   \entA A {\rho_{A}\otimes \identity_B} {\sigma_{A} \otimes \identity_B} = \ent {\rho_{A}} {\sigma_{A}}   $ 
\end{center}

In general, if $\mathcal{T}$ is a quantum channel, notice that the following holds,
\begin{center}
$\ds   \entA A { (\tr_B \circ \mathcal{T})( \rho_{AB} ) } { (\tr_B \circ \mathcal{T}) (\sigma_{AB} )} =  \ent { (\tr_B \circ \mathcal{T})( \rho_{AB} ) } { (\tr_B \circ \mathcal{T}) (\sigma_{AB} )}    $.
\end{center}
since $(\tr_B \circ \mathcal{T})( \rho_{AB} )$ and $ (\tr_B \circ \mathcal{T}) (\sigma_{AB} )$ have support in $A$.
\end{remark}

\begin{remark}
Notice that, by property (2.1), if $\rho_{AB}=\sigma_{AB}$, in particular $\rho_{AB}= \E^*_A( \rho_{AB} )$ and, thus, $\entA A {\rho_{AB}} {\sigma_{AB}} $=0, but the converse implication is false in general (differently from the case of the relative entropy). Both implications cannot hold, since that would be incompatible with property (3.1).
\end{remark}

Properties (1), (2), (3) and (3.1) come from the necessity that the conditional relative entropy extends the relative entropy. The names of properties (3) and (3.1) come from the fact that $D^+_{A,B}$ actually satisfies the properties of additivity and superadditivity.

Let us define

\begin{center}
$D_{A,B}^+(\rho_{AB} || \sigma_{AB}):= D_A (\rho_{AB} || \sigma_{AB}) + D_B (\rho_{AB} || \sigma_{AB})$.
\end{center} 

Then, we can prove a couple of properties for $D^+_{A,B}$ that yield the relation of this concept with the usual relative entropy.

\begin{prop}\label{prop:D+}
Let $\hs_{AB}=\hs_A \otimes \hs_B$ and $ {\rho_{AB}}, {\sigma_{AB}} \in \SSS_{AB}$. $D^+_{A,B}$ satisfies the following properties:
\begin{enumerate}
\item \textbf{Additivity:} $D^+_{A,B}({\rho_{A}\otimes \rho_B}||  {\sigma_{A} \otimes \sigma_B}  )=  \ent {\rho_{A}} {\sigma_{A}} +  \ent {\rho_{B}} {\sigma_{B}} $.
\item \textbf{Superadditivity:} $D^+_{A,B}({\rho_{AB}}||  {\sigma_{A} \otimes \sigma_B}  ) \geq \ent {\rho_{A}} {\sigma_{A}} +  \ent {\rho_{B}} {\sigma_{B}} $.
\end{enumerate}

\end{prop}

\begin{proof}
\begin{itemize}
\item (1) follows from property (3.1) in the definition of conditional relative entropy.
\item (2) is obtained from property (3) in the definition of conditional relative entropy.
\end{itemize}
\end{proof}

Notice from the previous definition that the main difference between the relative entropy and $D^+_{A,B}$ lies in the fact that the latter lacks the property of monotonicity. Indeed, as mentioned above, since $D^+_{A,B}$ verifies the properties of continuity, additivity and superadditivity, we know that it cannot verify the property of monotonicity (i.e., data processing for every quantum channel), as it would imply that it is a multiple of the relative entropy \cite{axcharRE}. This motivates the appearance of the property of ``semi-monotonicity".

To justify the name for  that property, let us comment a bit on every term of the inequality. Comparing the first term of both sides of the inequality, we can see a data-processing inequality for $D_A$. We know that such inequality cannot be true in general, since,  for the conditional relative entropy in $A$, a quantum channel with support in $B$ is not expected to decrease this quantity. This fact justifies the presence of the second term on both sides of the inequality, to compensate the non-decreasing effect of the  ``$B$-part" of a channel in the conditional relative entropy in $A$. Hence, we add the conditional relative entropy of this ``$B$-part" of the channel in $B$, where we know that the decreasing effect actually holds.  

In the following subsection we will show that there exists a unique  conditional relative entropy.

\subsection{A formula for the conditional relative entropy}

\begin{theorem}\label{thm:formula}
Let $D_A(\cdot || \cdot)$ be a conditional relative entropy, according to Definition \ref{def:CRE}. Then, $D_A(\cdot || \cdot)$ is explicitly given by  
\begin{center}
$\entA A {\rho_{AB}} {\sigma_{AB}} = \ent {\rho_{AB}} {\sigma_{AB}} - \ent {\rho_{B}} {\sigma_{B}}$,
\end{center}
for every $\rho_{AB}, \sigma_{AB} \in \SSS_{AB}$.
\end{theorem}

\begin{proof}

Let us first prove that the quantity $D_A$ fulfils all conditions in Definition \ref{def:CRE}. Recall that we need prove the following properties:

\begin{enumerate}
\item \textit{The map $\rho_{AB} \mapsto  \entA A {\rho_{AB}} {\sigma_{AB}}$ is continuous.}

\vspace{0.2cm}

It is clear that $\ent {\rho_{AB}} {\sigma_{AB}}$ is continuous in $\rho_{AB}$, so $\ent {\rho_{B}} {\sigma_{B}}$ also is. Hence, their difference is also continuous. 

\vspace{0.2cm}

\item \textit{$ \entA A {\rho_{AB}} {\sigma_{AB}} \geq 0 $  and 
\begin{enumerate}
\item[(2.1)] $\entA A {\rho_{AB}} {\sigma_{AB}} $=0 if, and only if, $\rho_{AB}= \E^*_A( \rho_{AB} )$.
\end{enumerate}}

\vspace{0.2cm}

Notice that (2) is the monotonicity of the relative entropy (property (3) of Proposition \ref{prop:REprop}) for the channel $\tr_A$, and Property (2.1) was proven in \cite{petzmre}. 

\vspace{0.2cm}

\item \textit{$ \entA A {\rho_{AB}} {\sigma_{A} \otimes \sigma_B} \geq \ent {\rho_{A}} {\sigma_{A}} $ and
\begin{enumerate}
\item[(3.1)]  if $\rho_{AB}=\rho_A \otimes \rho_B$,  $ \entA A {\rho_{A}\otimes \rho_B} {\sigma_{A} \otimes \sigma_B} = \ent {\rho_{A}} {\sigma_{A}}   $.
\end{enumerate}}

\vspace{0.2cm}

In (3), using the superadditivity of the relative entropy,
\begin{center}
$\ds  \entA A {\rho_{AB}} {\sigma_{A} \otimes \sigma_B} \geq \ent {\rho_{A}} {\sigma_{A}}  +\ent {\rho_{B}} {\sigma_{B}} -  \ent {\rho_{B}} {\sigma_{B}} =\ent {\rho_{A}} {\sigma_{A}}  $.
\end{center}

 For (3.1), we have equality in the previous inequality:
\begin{center}
$\ds  \entA A  {\rho_{A}\otimes \rho_B} {\sigma_{A} \otimes \sigma_B} = \ent {\rho_{A}} {\sigma_{A}}  +\ent {\rho_{B}} {\sigma_{B}} -  \ent {\rho_{B}} {\sigma_{B}} =\ent {\rho_{A}} {\sigma_{A}}  $.

\end{center}

\vspace{0.2cm}

\item \textit{For every quantum channel $\mathcal{T}$, the following holds:
\begin{align*}
 D_A(\mathcal{T}(\rho_{AB})|| \mathcal{T}(\sigma_{AB})) +&  D_B((\tr_A \circ \mathcal{T}) (\rho_{AB}) || (\tr_A \circ \mathcal{T}) (\sigma_{AB}) ) \\
& \leq D_A(\rho_{AB}|| \sigma_{AB}) + D_B( \tr_A(\rho_{AB}) || \tr_A(\sigma_{AB}) ),
\end{align*}
where $\ds D_B({\rho_{AB}} || {\sigma_{AB}} ) $ is the conditional relative entropy in $B$.}

\vspace{0.2cm}

The first term in the LHS is expressed as:
\begin{align*}
 \entA A {\mathcal{T}(\rho_{AB})}  {\mathcal{T}(\sigma_{AB})} 
& = \ent {\mathcal{T}(\rho_{AB})}  {\mathcal{T}(\sigma_{AB})} - D_B((\tr_A \circ \mathcal{T}) (\rho_{AB}) || (\tr_A \circ \mathcal{T}) (\sigma_{AB}) ).
\end{align*}

Hence, the LHS of the statement of the proposition is actually given by
\begin{align*}
 D_A(\mathcal{T}(\rho_{AB})|| \mathcal{T}(\sigma_{AB})) +&  D_B((\tr_A \circ \mathcal{T}) (\rho_{AB}) || (\tr_A \circ \mathcal{T}) (\sigma_{AB}) ) =  \ent {\mathcal{T}(\rho_{AB})}  {\mathcal{T}(\sigma_{AB})}.
\end{align*}

Now, for the first term in the RHS, we have
\begin{center}
$ D_A(\rho_{AB}|| \sigma_{AB})=  \ent {\rho_{AB}} {\sigma_{AB}} - \ent {\rho_{B}} {\sigma_{B}} $.
\end{center}

Thus, the RHS can be rewritten as
\begin{align*}
 D_A(\rho_{AB}|| \sigma_{AB}) + D_B( \tr_A(\rho_{AB}) || \tr_A(\sigma_{AB}) ) &=  D_A(\rho_{AB}|| \sigma_{AB}) + \ent {\rho_{B}} {\sigma_{B}} \\
 &=  \ent {\rho_{AB}} {\sigma_{AB}},
\end{align*}
where in the first line we have used Remark \ref{remark:recoverabilityRE}.

In conclusion, the statement of the property is equivalent to the following inequality:
\begin{center}
$\ent {\mathcal{T}(\rho_{AB})}  {\mathcal{T}(\sigma_{AB})} \leq \ent {\rho_{AB}} {\sigma_{AB}} $,
\end{center}
which holds for every quantum channel $\mathcal{T}$ (property of monotonicity in Proposition \ref{prop:REprop}).

\end{enumerate}

Once we have proven that this definition of $D_A$ is indeed a conditional relative entropy according to Definition \ref{def:CRE}, we can move forward to the proof of the converse implication.

Let us define $f : \SSS_{AB} \times \SSS_{AB} \rightarrow \R_0^+ $ by:
\begin{center}
$\ds f(\rho_{AB}, \sigma_{AB}) = D_A (\rho_{AB} || \sigma_{AB})+ D ( \rho_{B} || \sigma_{B})$,
\end{center}
where $D_A$ is a conditional relative entropy (and, hence, satisfies the properties of Definition \ref{def:CRE}).

We are going to prove that 
\begin{center}
$\ds  f (\rho_{AB}, \sigma_{AB}) = \ent {\rho_{AB}} {\sigma_{AB}} .$
\end{center}

In virtue of the characterization for the relative entropy shown in \cite{axcharRE}, we just need to prove that $f$ satisfies the following properties for every $\rho_{AB}, \sigma_{AB} \in \SSS_{AB}$:
\begin{enumerate}
\item \textit{\textbf{Continuity:} $\rho_{AB} \mapsto f(\rho_{AB}, \sigma_{AB}) $ is continuous.}

\vspace{0.2cm}

It is a direct consequence of property (1) in Definition \ref{def:CRE}.

\vspace{0.2cm}

\item \textit{\textbf{Additivity:} $f(\rho_A \otimes \rho_B, \sigma_A \otimes \sigma_B) = f(\rho_A, \sigma_A) + f(\rho_B, \sigma_B)$.}

\vspace{0.2cm}

This follows from property (3.1) in Definition \ref{def:CRE}.

\vspace{0.2cm}

\item \textit{\textbf{Superadditivity:} $f(\rho_{AB}, \sigma_A \otimes \sigma_B) \geq f(\rho_A, \sigma_A) + f(\rho_B, \sigma_B)$.}

\vspace{0.2cm}

It is straightforward from property (3) in Definition \ref{def:CRE}.

\vspace{0.2cm}

\item \textit{\textbf{Monotonicity:} For every quantum channel $\mathcal{T}$, 
\begin{center}
$ f(\mathcal{T}(\rho_{AB}), \mathcal{T}(\sigma_{AB})) \leq  f(\rho_{AB}, \sigma_{AB}).$
\end{center}}

\vspace{0.2cm}

Let us study the second term in the definition of $f$:
\begin{align*}
\ent {\rho_B} {\sigma_B}  & =   \entA B {\rho_B} {\sigma_B} \\
& =   \entA B {\tr_A[\rho_{AB}]} {\tr_A[\sigma_{AB}]},
\end{align*}
where we have used Remark \ref{remark:recoverabilityRE} in the second line.

Therefore, we can rewrite the property of monotonicity of $f$ as:
\begin{align*}
 D_A(\mathcal{T}(\rho_{AB})|| \mathcal{T}(\sigma_{AB})) +&  D_B((\tr_A \circ \mathcal{T}) (\rho_{AB}) || (\tr_A \circ \mathcal{T}) (\sigma_{AB}) ) \\
& \leq D_A(\rho_{AB}|| \sigma_{AB}) + D( \rho_{B} || \sigma_{B} ),
\end{align*}
and this property holds by assumption, because of the property of semi-monotonicity.

\end{enumerate}

Hence, recalling \cite{axcharRE}, we can deduce that
\begin{center}
$ f(\rho_{AB}, \sigma_{AB}) \propto \ent {\rho_{AB}} {\sigma_{AB}} $.
\end{center}

Moreover, if we take $\rho_{AB}= \rho_A \otimes \rho_B$ and $\sigma_{AB}= \sigma_A \otimes \sigma_B$, we have
\begin{align*}
f(\rho_A \otimes \rho_B, \sigma_A \otimes \sigma_B) & = D_A (\rho_A \otimes \rho_B ||\sigma_A \otimes \sigma_B)+ D ( \rho_{B} || \sigma_{B}) \\
&= \ent {\rho_A} {\sigma_A} +  \ent {\rho_B} {\sigma_B} \\
&= \ent {\rho_A \otimes \rho_B} {\sigma_A \otimes \sigma_B},
\end{align*}
from which we can conclude
\begin{center}
$ f(\rho_{AB}, \sigma_{AB}) = \ent {\rho_{AB}} {\sigma_{AB}} $.
\end{center}

This fact immediately yields the statement of the theorem.
\end{proof}

\begin{remark}
Throughout the whole paper we are assuming that all the states considered are full-rank, and, thus, their relative entropy is finite. Hence, the conditional relative entropy, which we have just seen that can be expressed as a difference of relative entropies, is the difference of two finite quantities, so it is always well-defined. 

\end{remark}

The formula obtained in this subsection for the conditional relative entropy allows us to give an operational interpretation to this quantity. In the context of thermodynamics and cost of quantum processes, in \cite{faist1}, the authors introduced the concept of coherent relative entropy to give a measure of the amount of information forgotten by a logical process, conditioned to the output of the process, and relative to certain weights encoded in an operator. In thermodynamics, this quantity can be seen as the work cost of a certain quantum process (some applications and interesting properties of this quantity have appeared in \cite{faist2}). 

Our conditional relative entropy coincides with the coherent relative entropy when the process considered is a partial trace and taking the i.i.d. limit, which allows us to think that the relative entropy might be of use in a wider range of physical and information-theoretic situations.

\subsection{Conditional relative entropy by expectations}\label{subsec:CREexp}

In the previous result we have shown that there exists a unique conditional relative entropy fulfilling all properties from Definition \ref{def:CRE}. However, the properties of continuity, non-negativity, semi-additivity and semi-superadditivity are expected to hold in such concept of conditional relative entropy, but the property of semi-monotonicity, although  justified below the definition, may seem less natural. 

One could then think of modifying, or just removing this property from the definition. That would leave space for more possible examples of this new \textit{modified conditional relative entropy}. The purpose of this subsection is, indeed, to introduce an example of a conditional relative entropy that lacks the property of semi-monotonicity.

One quantity widely used in quantum information theory is the relative entropy of a state and the Petz recovery map of the partial trace of that state. Indeed, it is known that there are cases where this quantity coincides with the aforementioned conditional relative entropy (this will be discussed in Subsection \ref{subsec:comparison}). Hence, it is also natural to study this quantity as a possible modified conditional relative entropy.

\begin{defi}[Conditional relative entropy by expectations]

Let $\hs_{AB}=\hs_A \otimes \hs_B$ be a composite Hilbert space and $\rho_{AB}, \sigma_{AB} \in S_{AB}$. Let $\E_A^*$ be the adjoint of the minimal conditional expectation introduced in subsection \ref{subsec-cond-exp}. We define the \textit{conditional relative entropy by expectations} of $\rho_{AB}$ and $\sigma_{AB}$ in $A$ by:
\begin{center}
$ \ds \enteA A {\rho_{AB}} {\sigma_{AB}} = \ent {\rho_{AB}} {\E_A^*(\rho_{AB})} $.
\end{center}
\end{defi}

We prove now that this definition fulfils all properties of conditional relative entropy, except for the semi-monotonicity.

\begin{prop}
Let $\hs_{AB}=\hs_A \otimes \hs_B$. The following properties hold for every $\rho_{AB}, \sigma_{AB} \in \SSS_{AB}$:
\begin{enumerate}
\item The map $\rho_{AB} \mapsto  \enteA A {\rho_{AB}} {\sigma_{AB}}$ is continuous.
\item $ \enteA A {\rho_{AB}} {\sigma_{AB}} \geq 0 $  and 
\begin{enumerate}
\item[(2.1)] $\enteA A {\rho_{AB}} {\sigma_{AB}} $=0 if, and only if, $\rho_{AB}= \E^*_A( \rho_{AB} )$.
\end{enumerate}
\item $ \enteA A {\rho_{AB}} {\sigma_{A} \otimes \sigma_B} \geq \ent {\rho_{A}} {\sigma_{A}} $ and
\begin{enumerate}
\item[(3.1)]  if $\rho_{AB}=\rho_A \otimes \rho_B$,  $ \enteA A {\rho_{A}\otimes \rho_B} {\sigma_{A} \otimes \sigma_B} = \ent {\rho_{A}} {\sigma_{A}}   $.
\end{enumerate}
\end{enumerate}
\end{prop}

\begin{proof}
\begin{itemize}
\item (1) is due to the facts that $\E^*_A(\rho_{AB})$ is linear in $\rho_{AB}$ and the relative entropy is continuous.
\item Property (2) comes from the fact that the conditional relative entropy by expectations is in particular a relative entropy of density matrices.
\item (2.1) is a consequence of the fact that the relative entropy of two states vanishes if, and only if, they coincide.
\item For (3), observe that if $\sigma_{AB}= \sigma_A \otimes \sigma_B$,
\begin{align*}
\E_{A}^*(\rho_{AB})&=  \sigma_{A}^{1/2} \otimes \sigma_B^{1/2} \sigma_B^{-1/2}\rho_B \sigma_B^{-1/2}  \sigma_{A}^{1/2} \otimes \sigma_B^{1/2} \\
&= \sigma_A \otimes \rho_B. 
\end{align*}
Hence,
\begin{align*}
 \enteA A {\rho_{AB}} {\sigma_{A} \otimes \sigma_B}&= \ent {\rho_{AB}} {\sigma_A \otimes \rho_B}\\
 &= \ent {\rho_{AB}} {\rho_A \otimes \rho_B} + \ent {\rho_{A}} {\sigma_A } \\
 &\geq \ent {\rho_{A}} {\sigma_A },
\end{align*}
where we have used the non-negativity of the relative entropy.
\item In (3.1), if both $\rho_{AB}$ and $\sigma_{AB}$ are tensor products, we have equality in the previous inequality:
\begin{align*}
\enteA A {\rho_{A}\otimes \rho_B} {\sigma_{A} \otimes \sigma_B}& = \ent {\rho_{A}\otimes \rho_B} {\sigma_{A}\otimes \rho_B}\\
& = \ent {\rho_{A}} {\sigma_{A}},
\end{align*}
since $ \ent {\rho_A \otimes \rho_B}{\sigma_A \otimes \rho_B}= \ent {\rho_A }{\sigma_A } + \ent { \rho_B}{\rho_B}$ and the second term is zero.
\end{itemize}
\end{proof}

%In the second property of this proposition again only one implication holds true, since we could have $\E_A^*(\rho_{AB}) = \rho_{AB}$ with $\rho_{AB} \neq \sigma_{AB}$. In fact, we will see in the following subsection that the situations for which both definitions of conditional relative entropy are null coincide.

\subsection{Comparison of definitions}\label{subsec:comparison} 

Once we have presented both definitions for conditional relative entropy and conditional relative entropy by expectations, respectively, it is a natural question whether they coincide in general and, if not, characterize the states for which they do. Let us consider $\rho_{AB}$ and $\sigma_{AB}$ bipartite density matrices in $\hs_{AB}= \hs_A \otimes \hs_B $ and study different cases.

\vspace{0.2cm}

\textbf{Case 1: $\rho_B$, $\sigma_{AB}$ and $\sigma_B$ commute.}

We first assume that $\left[ \rho_{B}, \sigma_{AB} \right]=\left[ \rho_{B}, \sigma_{B} \right] =\left[ \sigma_{B}, \sigma_{AB} \right] =0$. Then, we can rewrite both definitions of conditional relative entropies as:
\begin{eqnarray*}
\entA A {\rho_{AB}} {\sigma_{AB}} &=& \ent {\rho_{AB}} {\sigma_{AB}}  - \ent {\rho_{B}} {\sigma_{B}} \\
&=& \tr[\rho_{AB}(\log \rho_{AB} - \log \sigma_{AB} \rho_B \sigma_B^{-1})]
\end{eqnarray*}
and 
\begin{eqnarray*}
\enteA A {\rho_{AB}} {\sigma_{AB}} &=& \ent {\rho_{AB}} {\E^*_{A}(\rho_{AB})} \\
&=& \tr[\rho_{AB}(\log \rho_{AB} - \log \sigma_{AB} \rho_B \sigma_B^{-1})],
\end{eqnarray*}
so we can see that they coincide.

\vspace{0.2cm}

\textbf{Case 2: $\sigma_{AB}$ has the splitting property.}

Suppose that $\sigma_{AB}= \sigma_A \otimes \sigma_B$. Then, for the conditional relative entropy, we have
\begin{eqnarray*}
\entA A {\rho_{AB}} {\sigma_{A} \otimes \sigma_B} &=& \ent {\rho_{AB}} {\sigma_{A} \otimes \sigma_B} - \ent {\rho_B} {\sigma_B}\\
&=& \ent {\rho_{AB}} {\rho_A \otimes \rho_B}  + \ent {\rho_A \otimes \rho_B}    {\sigma_A \otimes \sigma_B}   -\ent {\rho_B} {\sigma_B}   \\
&=&  I_\rho(A:B) + \ent {\rho_A} {\sigma_A}+\ent {\rho_B} {\sigma_B} -\ent {\rho_B} {\sigma_B} \\
&=& I_\rho(A:B) + \ent {\rho_A} {\sigma_A}.
\end{eqnarray*}

Furthermore, the adjoint of the minimal conditional expectation takes the value $\E^*_A(\rho_{AB}) = \sigma_A \otimes \rho_B$. Thus, the conditional relative entropy by expectations in this case is given by
\begin{eqnarray*}
\enteA A {\rho_{AB}} {\sigma_{AB}} &=& \ent {\rho_{AB}} {\sigma_A \otimes \rho_B} \\
&=& \ent {\rho_{AB}} {\rho_A \otimes \rho_B} + \ent {\rho_A \otimes \rho_B}{\sigma_A \otimes \rho_B} \\
&=& I_\rho(A:B) + \ent {\rho_A} {\sigma_A}.
\end{eqnarray*}

Therefore,
\begin{center}
$ \entA A {\rho_{AB}} {\sigma_A \otimes \sigma_B} = \enteA A {\rho_{AB}} {\sigma_A \otimes \sigma_B} $.
\end{center}

\vspace{0.2cm}

\textbf{Case 3: $\entA A {\rho_{AB}} {\sigma_{AB}} = 0$ or $\enteA A {\rho_{AB}} {\sigma_{AB}} = 0$.}

On the one hand, for the conditional relative entropy by expectations, since it is in particular a relative entropy between two states it is clear that (property (1) of Proposition \ref{prop:REprop}):
\begin{center}
$\ds \enteA A {\rho_{AB}} {\sigma_{AB}} = 0 \, \Leftrightarrow \, \rho_{AB} = \E_A^*(\rho_{AB})$.
\end{center}

On the other hand, for the conditional relative entropy, the situation
\begin{center}
$\ds   \entA A {\rho_{AB}} {\sigma_{AB}} = 0  \, \Leftrightarrow \,  \ent {\rho_{AB}} {\sigma_{AB}} = \ent {\rho_{B}} {\sigma_{B}} $
\end{center}
was addressed and characterized by Petz in \cite{petzmre}. In general, if $\hs$ and $\mathcal{K}$ are two Hilbert spaces, $\mathcal{T}:\BB(\hs) \rightarrow \BB(\mathcal{K})$ is a quantum channel and $\rho$ and $\sigma$ are two states of $\hs$, the following inequality, known as the monotonicity of the relative entropy, holds (\cite{uhlmann}, previously proven for the finite-dimensional case in \cite{lindblad}):
\begin{equation}\label{eq:mon-rel-ent}
\ds \ent \rho \sigma \geq \ent {\mathcal{T}(\rho)} {\mathcal{T}(\sigma)} ,
\end{equation}
and Petz proved that (\ref{eq:mon-rel-ent}) is indeed saturated if, and only if, both $\rho$ and $\sigma$ can be recovered in the following way
\begin{center}
$\ds \widehat{\mathcal{T}} \mathcal{T} (\rho) = \rho  , \phantom{sad}\widehat{\mathcal{T}} \mathcal{T} (\sigma) = \sigma$,
\end{center} 
where $\widehat{\mathcal{T}} $ can be explicitly given by:

\begin{center}
$\ds \widehat{\mathcal{T}} \rho = \sigma^{1/2} \, \mathcal{T}^* \left( (\mathcal{T} \sigma)^{-1/2} \rho \, (\mathcal{T} \sigma)^{-1/2} \right) \sigma^{1/2} $.
\end{center}

Notice from the expression of $\widehat{\mathcal{T}} $ that $\widehat{\mathcal{T}} \mathcal{T} (\sigma) = \sigma$ always holds.

For the particular case of the partial trace, this problem was also addressed in \cite{ssequality}, where the authors showed that, for this channel, having equality in equation (\ref{eq:mon-rel-ent}) is equivalent to having equality in the strong subadditivity  of the relative entropy (Proposition \ref{prop:sigmaprod}), and, using this characterization they provided a decomposition into a direct sum of mutually orthogonal tensor products for the  states which satisfy strong subadditivity with equality.

In what concerns equality in (\ref{eq:mon-rel-ent}) for the partial trace, Petz' result reads as:
\begin{center}
$\ds \entA A {\rho_{AB}} {\sigma_{AB}} = 0 \, \Leftrightarrow \, \rho_{AB}= \sigma_{AB}^{1/2} \, \sigma_B^{-1/2} \, \rho_{B} \, \sigma_B^{-1/2} \,  \sigma_{AB}^{1/2}  = \E^*_A(\rho_{AB})$,
\end{center}
where we recall that $ \E^*_{A}(\rho_{AB})$ is exactly the Petz recovery map for the partial trace $\tr_A$.

Therefore, the kernels of both definitions of conditional relative entropies coincide, i.e., both vanish under the same conditions.

\vspace{0.2cm}

\textbf{Case 4: General case.}

We have seen that both definitions mentioned above coincide, at least, when they are null, $\sigma_{AB}$ is a tensor product, or $\left[ \rho_{B}, \sigma_{AB} \right]=\left[ \rho_{B}, \sigma_{B} \right] =\left[ \sigma_{B}, \sigma_{AB} \right] =0$. In general, as far as we know, the problem of characterizing for which states $\rho_{AB}, \sigma_{AB}$ both definitions coincide is still an open question.

Another natural question that arises in this context is whether one definition could be always greater or equal than the other, i.e., whether the following inequality
\begin{equation}\label{eq:two-defs}
\entA A {\rho_{AB}} {\sigma_{AB}} \ov{?}{\geq} \enteA A {\rho_{AB}} {\sigma_{AB}}
\end{equation}
or the reverse one hold for every $\rho_{AB}, \sigma_{AB} \in \SSS_{AB}$. The left-hand side of this inequality has been widely studied in a series of recent papers (\cite{fawzirenner}, \cite{bertamre}, \cite{swiveled}, \cite{sutterur}, \cite{sutter}, among other results), where the authors provide several lower and upper bounds for our conditional relative entropy by differences.

This result is already known to be false. Let us consider a tripartite Hilbert space $\hs_{ABC}=\hs_A \otimes \hs_B \otimes \hs_C $ and compare both definitions of conditional relative entropy in $C$. Consider $\rho_{ABC} \in \SSS_{ABC}$ and suppose that $\sigma_{ABC}=\identity_A \otimes \rho_{BC}$. Then,
\begin{align*}
\entA C {\rho_{ABC}} {\sigma_{ABC}} & =  \ent {\rho_{ABC}} {\sigma_{ABC}} - \ent{\rho_{AB}} {\sigma_{AB}} \\
& = \ent {\rho_{ABC}} {\rho_{BC}} - \ent{\rho_{AB}} {\rho_{B}} \\
& =  - S[\rho_{ABC}] + S[\rho_{BC}] + S[\rho_{AB}] - S[\rho_B] \\
& = I_\rho (A:BC) - I_\rho (A:B) \\
&= I_\rho(A:C | B),
\end{align*}
where this last term is called \textit{conditonal mutual information}, and
\begin{align*}
\enteA C {\rho_{ABC}} {\sigma_{ABC}} & =  \ent {\rho_{ABC}} {\E^*_C(\rho_{ABC})} \\
& = \ent {\rho_{ABC}^{\phantom{s}}} {\rho_{BC}^{1/2} \, \rho_{B}^{-1/2} \, \rho_{AB}\, \rho_{B}^{-1/2} \, \rho_{BC}^{1/2}}. 
\end{align*}

Hence, inequality (\ref{eq:two-defs}) in the particular case $\sigma_{ABC}=\identity_A \otimes \rho_{BC}$ can be rewritten as 
\begin{equation}\label{eq:two-defs-part}
I_\rho(A:C | B) \ov{?}{\geq} \ent {\rho_{ABC}} {\rho_{BC}^{1/2} \, \rho_{B}^{-1/2} \, \rho_{AB}\, \rho_{B}^{-1/2} \, \rho_{BC}^{1/2}}.
\end{equation}

This problem was addressed in \cite{conjfalse1}, where they considered these two quantities and plotted one quantity against the other for 10.000 randomly chosen pure states of dimension $2 \times 2 \times 2$. They showed  that even though in most of the cases the conditional mutual information is strictly greater than the conditional relative entropy by expectations, there are cases in which the reverse inequality holds. Similar numerical results had also been obtained in \cite{li-winter}.

In the recent paper \cite{conjfalse2}, the authors studied the following inequality:
\begin{equation}\label{eq:two-defs-part2}
I_\rho(A:C | B) \ov{?}{\geq} \un{\Lambda: B \rightarrow BC}{\text{min}}  \ent {\rho_{ABC}} {(\identity_A \otimes \Lambda	)(\rho_{AB})}.
\end{equation}

They tested it on 2.000 randomly chosen pure states of dimension $2 \times 2 \times 2$ and showed that there are states for which inequality (\ref{eq:two-defs-part2}) is violated. For these states, in particular, inequality (\ref{eq:two-defs-part}) is also violated, since

\begin{center}
$\ds  I_\rho(A:C | B) < \un{\Lambda: B \rightarrow BC}{\text{min}}  \ent {\rho_{ABC}} {(\identity_A \otimes \Lambda	)(\rho_{AB})} \leq \ent {\rho_{ABC}} {\rho_{BC}^{1/2} \, \rho_{B}^{-1/2} \, \rho_{AB}\, \rho_{B}^{-1/2} \, \rho_{BC}^{1/2}} $.
\end{center}

After that, they also presented an explicit counterexample for inequality (\ref{eq:two-defs-part2}).

\begin{remark}
A natural question in this context is whether one can recover the conditional entropy of $\rho_{AB}$ when one considers $\sigma_{AB}=\identity_{AB}$ in the conditional relative entropy, analogously to what happens for the von Neumann entropy of $\rho_{AB}$, which is recovered from the relative entropy of $\rho_{AB}$ and $\sigma_{AB}$ when  $\sigma_{AB}=\identity_{AB}$. 

More specifically, given a bipartite Hilbert space and $\rho_{AB}$ a state on it, for the conditional relative entropy in $A$ of $\rho_{AB}$ and $\identity_{AB}$ we have:
\begin{align*}
D_A(\rho_{AB} || \identity_{AB})  &= D(\rho_{AB} || \identity_{AB}) - D(\rho_{B} || d_A \identity_{B}) \\
&= - S(\rho_{AB}) + S(\rho_B) + \tr[\rho_B \log d_A] \\
&= - S(A|B)_{\rho} + \log d_A.
\end{align*}

Moreover, for the conditional relative entropy by expectations in $A$, we have
\begin{align*}
D_A^E(\rho_{AB} || \identity_{AB}) &=  D(\rho_{AB} || \identity_{AB}(d_A \identity_B)^{-1/2} \rho_B (d_A \identity_B)^{-1/2}   \identity_{AB}) \\
&= D(\rho_{AB} || \identity_{A}/d_A \otimes \rho_B ) \\
&= - S(A|B)_{\rho} + \log d_A.
\end{align*}

Hence, from both definitions we can recover the conditional entropy of $\rho$ in $A$ plus an additive factor with the logarithm of the dimension of $\mathcal{H}_A$, due to the fact that both definitions of conditional relative entropies were provided for states, instead of observables. If we compute both conditional relative entropies of $\rho_{AB}$ and $\identity_{AB}/d_{AB}$ (because now they are both states), then we recover the conditional entropy of $\rho_{AB}$ in both situations.

\end{remark}

\subsection{Relation with the classical case}\label{subsec:classical}

In this subsection, we will prove that both definitions presented above extend their classical analogue. Before that, we recall the classical definition for the entropy and the conditional entropy of a function, respectively.

Consider a probability space $(\Omega, \mathcal{F}, \mu)$ and define, for every $f>0$, the \textit{entropy} of $f$ by
\begin{center}
$\text{Ent}_\mu (f):= \mu(f \log f ) - \mu(f) \log \mu(f)$.  
\end{center}

Given a sub-$\sigma$-algebra $\mathcal{G}\subseteq \mathcal{F}$, we define the \textit{conditional entropy} of $f$ in $\mathcal{G}$ by
\begin{center}
$\text{Ent}_\mu (f \mid \mathcal{G}):= \mu(f \log f \mid \mathcal{G}) - \mu(f\mid \mathcal{G}) \log \mu(f\mid \mathcal{G})$,
\end{center}
where $\mu(f \mid \mathcal{G})$ is given by
\begin{center}
$\ds  \int_G \mu(f \mid \mathcal{G}) \, d \mu = \int_G f d \mu  \, \, $ for each $G \in \mathcal{G}$.
\end{center}

Let us consider two measures $\nu$ and $\mu$ in $(\Omega, \mathcal{F})$. We define the \textit{relative entropy} of $\nu$ with respect to $\mu$ by 
\begin{center}
$\ds H(\nu | \mu) := \begin{cases} \mu(f \log f) \, \text{ if } d \nu = f  d \mu, \, f \log f \in L^1 (\mu), \\
+ \infty \text{ \phantom{adsdr} otherwise}.
\end{cases} $
\end{center}

Then, we can relate it to the previous concept by
\begin{center}
$\ds H(\nu | \mu) = \text{Ent}_\mu \left( \frac{d \nu}{d \mu} \right) $,
\end{center}
and analogously to the definition of conditional entropy, we could define the \textit{conditional relative entropy} of $\nu$ with respect to $\mu$ in $\mathcal{G}$ by:

\begin{center}
$\ds  H_{\mathcal{G}}(\nu | \mu) = \text{Ent}_\mu (f | \mathcal{G}) $
\end{center}
for $\ds  f= \frac{d \nu}{d \mu}$.

Let us compare now this setting to the quantum case. We will prove  that, when the states are classical, both the conditional relative entropy and the conditional relative entropy by expectations coincide with the measure of the classical conditional entropy. The measure is necessary due to the fact that the classical conditional entropy of a function is another function, whereas the conditional relative entropy of two states produces a scalar.

We first rewrite the classical conditional entropy as:
\begin{align*}
\text{Ent}_\mu (f \mid \mathcal{G})&= \mu(f \log f \mid \mathcal{G}) - \mu(f\mid \mathcal{G}) \log \mu(f\mid \mathcal{G})\\
&= \mu(f \log f \mid \mathcal{G}) - \mu(f \log \mu(f\mid \mathcal{G})  \mid \mathcal{G}) \\
&=  \mu(f (\log f -  \log \mu(f\mid \mathcal{G}) ) \mid \mathcal{G}).
\end{align*}

Now, since $\mu ( \mu(\cdot| \mathcal{G}) )= \mu(\cdot)$,
\begin{equation}\label{eq:clas-cond-ent}
\mu( \text{Ent}_\mu (f \mid \mathcal{G}) ) = \mu(f (\log f -  \log \mu(f\mid \mathcal{G}) ) )
\end{equation}

Let us consider a bipartite Hilbert space $\hs_{AB}=\hs_A \otimes \hs_B $ and a classical state on it, i.e., of the form:

\begin{center}
$\ds \rho_{AB}= \un{a,b}{\sum} P_{AB}(a,b) \dyad{a}{a}_A \otimes \dyad{b}{b}_B  $.
\end{center}

Then, since the space of observables for each system is an abelian $C^*$-algebra, in virtue of Gelfand's theorem (see \cite{gelfand}, for instance) the composite system of observables can be expressed as
\begin{center}
$C(K) \otimes C(L) = C(K \times L)$,
\end{center}
where both $K$ and $L$ are compact spaces. A state in the composite system is a positive $\rho$ of the dual of $C(K \times L)$, which by the Riesz-Markov theorem (\cite{riesz}, \cite{markov}) corresponds to a regular Borel measure on $K \times L$. Hence, we can identify a classical state $\rho_{AB}$ with a regular measure $\mu$.

\begin{figure}
\begin{center}
\includegraphics[scale=0.42]{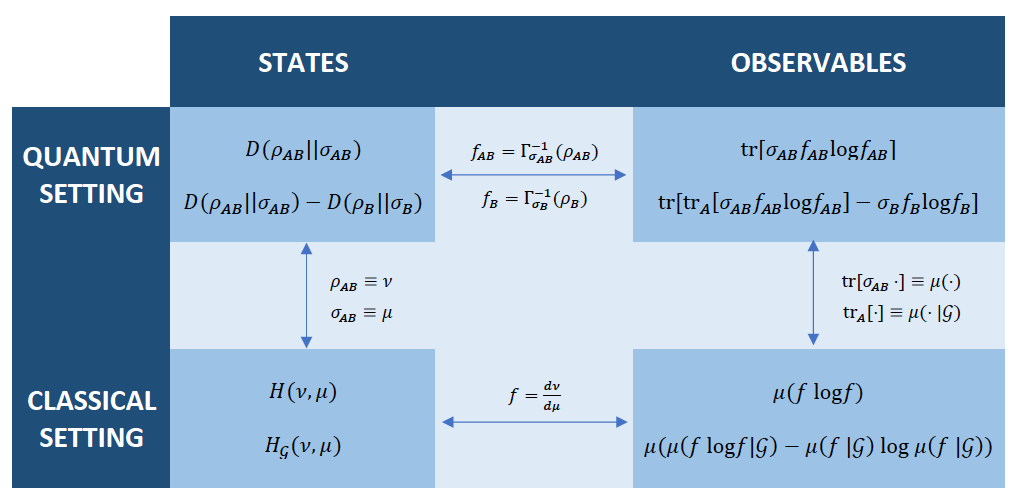}
\end{center}
\caption{Identification between classical and quantum quantities when the states considered are classical.}
\label{fig:identification}
\end{figure}

Moreover, we obtain the corresponding reduced state in one of the components by projecting the measure of $\rho_{AB}$ to that component, so the partial trace of the quantum setting can be interpreted as this operation in the classical setting (which is exactly the operation of the conditioning to a sub-$\sigma$-algebra in the definition of the classical conditional entropy). Thus, we identify $\tr_A[\cdot]$ with $\mu( \cdot | \mathcal{F})$.

Let us also recall that in the quantum setting we are considering states and in the definition of classical entropy, observables. The transition from the Schrödinger picture to the Heisenberg picture can be made by means of the operator:

\begin{center}
$\ds \Gamma^{-1}_{\sigma_{AB}} (\rho_{AB})= {\sigma_{AB}}^{-1/2}\rho_{AB}   \sigma_{AB}^{-1/2} $
\end{center}
for a certain full-rank state $\sigma_{AB}$, which we also consider classical. In particular, $\rho_{AB}$ and $\sigma_{AB}$ commute, as well as their marginals.

If we put all this together (a diagram of this identification can be seen in Figure \ref{fig:identification}), along with equation (\ref{eq:clas-cond-ent}), and identify the trace with respect to $\sigma_{AB}$ with the measure $\mu$, taking into account that $f=d\nu / d \mu$ is identified with $\Gamma^{-1}_{\sigma_{AB}} (\rho_{AB})$, we have:
\begin{align*}
\mu( \text{Ent}_\mu (f \mid \mathcal{G}) )& = \mu(f (\log f -  \log \mu(f\mid \mathcal{G}) ) )\\
&= \tr[\sigma_{AB} \Gamma^{-1}_{\sigma_{AB}} (\rho_{AB}) ( \log \Gamma^{-1}_{\sigma_{AB}} (\rho_{AB}) - \log \tr_A[\Gamma^{-1}_{\sigma_{AB}} (\rho_{AB})])] \\
&= \tr[\rho_{AB}(\log \rho_{AB} \sigma_{AB}^{-1} - \log \rho_{B} \sigma_{B}^{-1})] \\
&= \tr[\rho_{AB} (\log \rho_{AB}   - \log\sigma_{AB} - \log \rho_{B} + \log  \sigma_{B})]\\
&= \entA A {\rho_{AB} } {\sigma_{AB}}= \enteA A {\rho_{AB} } {\sigma_{AB}},
\end{align*}
where we have proven that both the quantum conditional relative entropy and the conditional relative entropy by expectations coincide with the measure of the classical conditional entropy.

\section{Quasi-factorization of the relative entropy}\label{sec:qf}

A result of \textit{quasi-factorization} of the relative entropy is an upper bound for the relative entropy of two states in terms of the sum of some conditional relative entropies, according to the definition of the previous section, in certain subregions, and a multiplicative error term which depends only on certain properties of the second state. The motivation for such results comes from classical spin systems, where a result of quasi-factorization of the classical entropy of a function, proven in both \cite{cesi} and \cite{clasico}, was essential for the simplification of a seminal result of \cite{marti-oliv} connecting the mixing time of some Glauber dynamics with the decay of correlations in their Gibbs states, via a positive log-Sobolev constant  (for more information on this topic, consult Section \ref{sec:logSob}).  

In this subsection, we present several quasi-factorization results for the relative entropy in terms of conditional relative entropies. We will classify these results in two classes, depending on the number of subregions where we condition and their overlap.

\begin{figure}
\begin{center}
\includegraphics[scale=0.8]{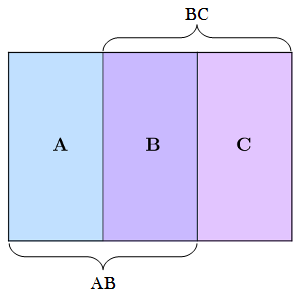}
\end{center}
\caption{Choice of indices in a tripartite Hilbert space $\hs_{ABC}= \hs_A \otimes \hs_B \otimes \hs_C$.}
\label{fig:2}
\end{figure}

In results from the first class, the relative entropy is bounded by the sum of two conditional relative entropies in overlapping regions and a multiplicative error term. Namely, for a tripartite Hilbert space 
$\hs_{ABC}= \hs_A \otimes \hs_B \otimes \hs_C$, we focus on subregions $AB$ and $BC$ (see figure \ref{fig:2}) and prove results of the kind 
\begin{equation}\label{eq:QF-Ov}
\tag{QF-Ov}
\boxed{(1-\xi(\sigma_{ABC}))\ent{\rho_{ABC}}{\sigma_{ABC}} \le \entA {AB} {\rho_{ABC}} {\sigma_{ABC}} + \entA {BC} {\rho_{ABC}} {\sigma_{ABC}}},
\end{equation}
for every $\rho_{ABC}, \sigma_{ABC} \in \SSS_{ABC}$, where $\xi(\sigma_{ABC})$ depends only on $\sigma_{ABC}$ and measures how far $\sigma_{AC}$ is from $\sigma_A \otimes \sigma_C$.

Results of this class constitute quantum analogues to Lemma \ref{lemma:clasico}. We will show below, and in Section \ref{subsec:QF-CRE}, some examples. 

For the second class of results of quasi-factorization, we assume that the regions where we are conditioning the relative entropy in the RHS do not overlap. Thus, imposing strong conditions on the second state, we are able to obtain  quasi-factorization results conditioning the relative entropy to a bigger number of regions. More specifically, for a $n$-partite Hilbert space 
$\hs_{A_1 \ldots A_n}= \un{i=1}{\ov{n}{\bigotimes}} \hs_{A_i} $, we prove results of the kind 
\begin{equation}\label{eq:QF-NonOv}
\tag{QF-NonOv}
\boxed{(1-\xi(\sigma_{A_1 \ldots A_n}))\ent{\rho_{A_1 \ldots A_n}}{\sigma_{A_1 \ldots A_n}} \le  \un{i=1}{\ov{n}{\sum}} \entA {A_i} {\rho_{A_1 \ldots A_n}} {\sigma_{A_1 \ldots A_n}}},
\end{equation}
for every $\rho_{A_1 \ldots A_n}, \sigma_{A_1 \ldots A_n} \in \SSS_{A_1 \ldots A_n}$, where $\xi(\sigma_{A_1 \ldots A_n})$ depends only on $\sigma_{A_1 \ldots A_n}$ and measures in some way how far it is from being a tensor product.

Some examples of results of this class will be presented in this subsection and subsequent ones. Indeed, the main result of Subsection \ref{subsec:QF-sigmaprod} will be used in the next section in quantum spin lattices, as the key step to prove positivity of a log-Sobolev constant.

\begin{remark}
It is clear that, whenever one has a result of the first class (\ref{eq:QF-Ov}), one can construct another one of the second class (\ref{eq:QF-NonOv}), conditioning in the RHS only in two regions, just by assuming that dim($\hs_B$)=1 in the first result. 
\end{remark}

\begin{remark}
In this and the following subsection, we will assume that $\sigma$ is always a tensor product, and, in that scenario, we have seen in Subsection \ref{subsec:comparison} that the definitions of conditional relative entropy and conditional relative entropy by expectations coincide. Hence, we can simply state that, except for Subsection \ref{subsec:QF-CREexp}, all the results of quasi-factorization presented in this paper concern conditional relative entropies, according to Definition \ref{def:CRE}.
\end{remark}

We will now present some results of quasi-factorization and classify them into the two classes mentioned above. Let us separate the study in different cases in increasing order of difficulty. First, we start showing some results that can be proven directly from the properties in the axiomatic definition of conditional relative entropy. When both states are products, we have the following possibilities:

\begin{enumerate}
\item \underline{$\text{dim}(\hs_B)=1$, $\rho_{AC}= \rho_A \otimes \rho_C$ and $\sigma_{AC}= \sigma_A \otimes \sigma_C$:}

\vspace{0.2cm}

From the property of additivity of Proposition \ref{prop:D+}, we can see that
\begin{center}
$\ds D^+_{A,C} ({\rho_{A}\otimes \rho_C} ||{\sigma_{A} \otimes \sigma_C}) = \ent {\rho_{A}} {\sigma_{A}}  + \ent {\rho_{C}} {\sigma_{C}} = \ent   {\rho_{A}\otimes \rho_C} {\sigma_{A}\otimes \sigma_C}.$
\end{center}

Hence, in this case,
\begin{center}
$\ent{\rho_{AC}}{\sigma_{AC}} = \entA {A} {\rho_{AC}} {\sigma_{AC}} + \entA {C} {\rho_{AC}} {\sigma_{AC}}$
\end{center}

constitutes the simplest result of quasi-factorization of both (\ref{eq:QF-Ov}) and (\ref{eq:QF-NonOv}).

\vspace{0.2cm}

\item \underline{Arbitrary dimension of $\hs_B$, $\rho_{ABC}= \rho_A \otimes \rho_B \otimes \rho_C$ and $\sigma_{ABC}= \sigma_A \otimes \sigma_B \otimes \sigma_C$:}

\vspace{0.2cm}

This case is an extension of the previous one.  We have  
\begin{center}
$ \ds \ent{\rho_{ABC}}{\sigma_{ABC}} = \entA {A} {\rho_{ABC}} {\sigma_{ABC}}+\entA {B} {\rho_{ABC}} {\sigma_{ABC}} + \entA {C} {\rho_{ABC}} {\sigma_{ABC}},  $
\end{center}

which is clearly a result of (\ref{eq:QF-NonOv}).

\vspace{0.2cm}

\item \underline{In general, for $n \in \N$, $\hs_{A_1 \ldots A_n}= \un{i=1}{\ov{n}{\bigotimes}} \hs_{A_i} $, $\rho_{A_1 \ldots A_n}= \un{i=1}{\ov{n}{\bigotimes}} \rho_{A_i}$ and $\sigma_{A_1 \ldots A_n}= \un{i=1}{\ov{n}{\bigotimes}} \sigma_{A_i}$ :}

\vspace{0.2cm}

This case is a generalization of the previous one.  Because of the property of semi-additivity, we clearly have  
\begin{center}
$ \ds \ent{\rho_{A_1 \ldots A_n}}{\sigma_{A_1 \ldots A_n}} =  \un{i=1}{\ov{n}{\sum}} \entA {A_i} {\rho_{A_1 \ldots A_n}} {\sigma_{A_1 \ldots A_n}},  $
\end{center}

which is a result of (\ref{eq:QF-NonOv}).

\vspace{0.2cm}

\item \underline{The regions $AB$ and $BC$, $\rho_{ABC}= \rho_A \otimes \rho_B \otimes \rho_C$ and $\sigma_{ABC}= \sigma_A \otimes \sigma_B \otimes \sigma_C$:} 

\vspace{0.2cm}

Under these assumptions we have
\begin{center}
$\ds D^+_{AB,BC} ({\rho_{ABC}} ||{\sigma_{ABC}}) = \ent {\rho_{A}} {\sigma_{A}}  + 2 \ent {\rho_{B}} {\sigma_{B}}+\ent {\rho_{C}} {\sigma_{C}} \geq \ent   {\rho_{ABC}} {\sigma_{ABC}},$
\end{center}

where the last inequality comes from the additivity and non-negativity of the relative entropy. Hence, 
\begin{center}
$\ds  \ent{\rho_{ABC}}{\sigma_{ABC}} \le \entA {AB} {\rho_{ABC}} {\sigma_{ABC}} + \entA {BC} {\rho_{ABC}} {\sigma_{ABC}} $
\end{center}

is a result of (\ref{eq:QF-Ov}).

\end{enumerate}
 
 \begin{remark}
Notice that, in the previous four cases, the term $\xi(\sigma_{ABC})$ does not appear in the quasi-factorization result. This is something reasonable, since this term should measure how far $\sigma_{AC}$ is from $\sigma_A \otimes \sigma_C$ and, by assumption, in this case, this `distance' is zero.   
 \end{remark}

Let us now consider again a tripartite Hilbert space, relax the assumption on $\rho_{ABC}$ and assume only  $\sigma_{ABC}= \sigma_A \otimes \sigma_B \otimes \sigma_C$. Without imposing any condition on $\rho_{ABC}$, we are not able to obtain results of quasi-factorization from properties (1)-(3) in Definition \ref{def:CRE} (and, thus, fulfilled by both the conditional relative entropy and the conditional relative entropy by expectations) as we have just done above. 

However, for the conditional relative entropy  (and  the conditional relative entropy by expectations), we have the following property: If $\sigma_{ABC}= \sigma_A \otimes \sigma_B \otimes \sigma_C$, it is easy to check that the following holds:
\begin{equation}\label{eq:prop-sigma-prod}
\entA A {\rho_{ABC}} {\sigma_{ABC}} =   I_\rho (A : BC) + \ent {\rho_A} {\sigma_A} .
\end{equation}

Indeed, it was explicitely proven in Subsection \ref{subsec:comparison}.

Taking into account this property, for subregions  $AB$ and $BC$, we present another quasi-factorization result of the kind (\ref{eq:QF-Ov}).

\begin{prop}\label{prop:part-qf}
Let $\hs_{ABC}= \hs_A \otimes \hs_B \otimes \hs_C$ and $\rho_{ABC}, \sigma_{ABC} \in \SSS_{ABC}$ such that $\sigma_{ABC}= \sigma_A \otimes \sigma_B \otimes \sigma_C$. The following inequality holds:
\begin{equation}
\ent{\rho_{ABC}}{\sigma_{ABC}} \leq \entA {AB} {\rho_{ABC}} {\sigma_{ABC}}+\entA {BC} {\rho_{ABC}} {\sigma_{ABC}} .
\end{equation}
\end{prop}

\begin{proof}
Due to property (\ref{eq:prop-sigma-prod}),  for both definitions we have:
\begin{align*}
& D^+_{AB,BC}(\rho_{ABC} || \sigma_A \otimes \sigma_B \otimes \sigma_C)  \\
& \phantom{sdsd} = \entA {AB} {\rho_{ABC}} {\sigma_A \otimes \sigma_B \otimes \sigma_C} +  \entA {BC} {\rho_{ABC}} {\sigma_A \otimes \sigma_B \otimes \sigma_C} \\
& \phantom{sdsd} = I_\rho(AB:C) + \ent {\rho_{AB}} {\sigma_{A}\otimes \sigma_{B}} + I_\rho(BC:A) + \ent {\rho_{BC}} {\sigma_{B}\otimes \sigma_{C}}. 
\end{align*}

Now, because of monotonicity of the relative entropy with respect to the partial trace and additivity,
\begin{align*}
\ds \ent {\rho_{AB}} {\sigma_{A}\otimes \sigma_{B}} + \ent {\rho_{BC}} {\sigma_{B}\otimes \sigma_{C}}& \geq \ent {\rho_{A}} {\sigma_{A}} + \ent {\rho_{BC}} {\sigma_{B}\otimes \sigma_{C}} \\
 &\geq \ent {\rho_{A} \otimes \rho_{BC}} {\sigma_{A} \otimes \sigma_{B}\otimes \sigma_{C}},
\end{align*}
and adding this term to $I_\rho(BC:A)$, we have:
\begin{align*}
I_\rho(BC:A) & + \ent {\rho_{A} \otimes \rho_{BC}} {\sigma_{A} \otimes \sigma_{B}\otimes \sigma_{C}} \\
& =   \ent {\rho_{ABC}} {\rho_{A} \otimes \rho_{BC}}  + \ent {\rho_{A} \otimes \rho_{BC}} {\sigma_{A} \otimes \sigma_{B}\otimes \sigma_{C}}\\
& = \ent {\rho_{ABC}} {\sigma_A \otimes \sigma_B \otimes \sigma_C}.
\end{align*}

Therefore, 
\begin{align*}
&D^+_{AB,BC}(\rho_{ABC} || \sigma_A \otimes \sigma_B \otimes \sigma_C)  \geq \\
&\phantom{asasasasaaaa} \geq  I_\rho(AB:C) + I_\rho(BC:A)  + \ent {\rho_{A} \otimes \rho_{BC}} {\sigma_{A} \otimes \sigma_{B}\otimes \sigma_{C}} \\
& \phantom{asasasasaaaa} =  I_\rho(AB:C) + \ent {\rho_{ABC}} {\sigma_A \otimes \sigma_B \otimes \sigma_C} \\
& \phantom{asasasasaaaa} \geq \ent {\rho_{ABC}} {\sigma_A \otimes \sigma_B \otimes \sigma_C} .
\end{align*}
\end{proof}

We will show a more general version of this proposition, when $\sigma$ is not a tensor product, in Subsection \ref{subsec:QF-CRE}.

Considering now the regions $A$, $B$ and $C$, and again due to property (\ref{eq:prop-sigma-prod}), we can prove the following result of (\ref{eq:QF-NonOv}).
\begin{prop}\label{prop:part-qf2}
Let $\hs_{ABC}= \hs_A \otimes \hs_B \otimes \hs_C$ and $\rho_{ABC}, \sigma_{ABC} \in \SSS_{ABC}$ such that $\sigma_{ABC}= \sigma_A \otimes \sigma_B \otimes \sigma_C$. The following inequality holds:
\begin{equation}
\ent{\rho_{ABC}}{\sigma_{ABC}} \leq \entA {A} {\rho_{ABC}} {\sigma_{ABC}}+\entA {B} {\rho_{ABC}} {\sigma_{ABC}} + \entA {C} {\rho_{ABC}} {\sigma_{ABC}}.
\end{equation}
\end{prop}
\begin{proof}

Analogously to the definition of $D^+_{A,B}$, we define $D^+_{A,B,C}$:
\begin{align*}
& D^+_{A,B,C}(\rho_{ABC} || \sigma_A \otimes \sigma_B \otimes \sigma_C):= \\
& \phantom{sd}= \entA A {\rho_{ABC}} {\sigma_A \otimes \sigma_B \otimes \sigma_C} +  \entA B {\rho_{ABC}} {\sigma_A \otimes \sigma_B \otimes \sigma_C} + \entA C {\rho_{ABC}} {\sigma_A \otimes \sigma_B \otimes \sigma_C} \\
 & \phantom{sd}= I_\rho(A:BC) + \ent {\rho_{A}} {\sigma_{A}} + I_\rho(B:AC) + \ent {\rho_{B}} {\sigma_{B}} +  I_\rho(C:AB) + \ent {\rho_{C}} {\sigma_{C}}.
\end{align*}
Now, we have the following lower bound for the mutual informations:
\begin{align*}
& I_\rho(A:BC)+ I_\rho(B:AC)+  I_\rho(C:AB) = \\ 
& \phantom{as}=\tr  \left[ \rho_{ABC} \left(3 \log \rho_{ABC} - \log {\rho_A \otimes \rho_{BC}}  -\log {\rho_B \otimes \rho_{AC}} - \log {\rho_C \otimes \rho_{AB}} \right) \right] \\
& \phantom{as}=\tr  \left[ \rho_{ABC} \left(3 \log \rho_{ABC} - \log {\rho_A \otimes \rho_{B}\otimes \rho_{C}} - \log { \rho_{BC}}  -\log { \rho_{AC}} - \log { \rho_{AB}} \right) \right] \\
& \phantom{as}= \ent {\rho_{ABC}} {\rho_A \otimes \rho_B \otimes \rho_C} + \tr  \left[ \rho_{ABC} \left(2 \log \rho_{ABC} - \log { \rho_{BC}}  -\log { \rho_{AC}} - \log { \rho_{AB}} \right) \right] \\
& \phantom{as}= \ent {\rho_{ABC}} {\rho_A \otimes \rho_B \otimes \rho_C} + \tr  \left[ \rho_{ABC} \left( \log \rho_{ABC} - \log { \rho_{BC}}  -\log { \rho_{AC}} +\log \rho_C  \right) \right] \\
& \phantom{assd} + \tr[\rho_{ABC} ( \log \rho_{ABC} - \log \rho_{AB} \otimes \rho_C)] \\
&  \phantom{as} \geq  \ent {\rho_{ABC}} {\rho_A \otimes \rho_B \otimes \rho_C}+ I_{\rho}(AB:C) \\
&  \phantom{as} \geq  \ent {\rho_{ABC}} {\rho_A \otimes \rho_B \otimes \rho_C},
\end{align*}
where we have used strong subadditivity for the von Neumann entropy and non-negativity for the relative entropy.

Therefore,
\begin{align*}
& D^+_{A,B,C}(\rho_{ABC} || \sigma_A \otimes \sigma_B \otimes \sigma_C) \\
& \phantom{sdsd} \geq \ent {\rho_{ABC}} {\rho_A \otimes \rho_B \otimes \rho_C} + \ent {\rho_{A}} {\sigma_{A}} +\ent {\rho_{B}} {\sigma_{B}}  + \ent {\rho_{C}} {\sigma_{C}} \\
& \phantom{sdsd} = \ent {\rho_{ABC}} {\sigma_A \otimes \sigma_B \otimes \sigma_C}.
\end{align*}
\end{proof}

If we consider two non-overlapping subregions instead of three in the RHS, the quasi-factorization result of Subsection \ref{subsec:QF-CREexp} constitutes a generalization of this proposition, when $\sigma$ is not a tensor product, for the conditional relative entropy by expectations. Moreover, if in the quasi-factorization result of Subsection \ref{subsec:QF-CRE} we assume $\text{dim}(\hs_B)=1$, that result is also a generalization  of this proposition for two subregions when $\sigma$ is not necessarily a tensor product, for the conditional relative entropy. In both results, we will need the explicit expressions of conditional relative entropy and conditional relative entropy by expectations, respectively, oppositely to the cases mentioned above, where we obtained quasi-factorization results just from some properties of the definitions.

Concerning the number of subregions (and, thus, number of conditional relative entropies in the RHS), this proposition can be generalized to $n$-partite Hilbert spaces. We will show that in the following subsection.

\subsection{General quasi-factorization for $\sigma$ a tensor product}\label{subsec:QF-sigmaprod}

In this subsection, we show that, imposing strong conditions on the second state, we manage to prove a quasi-factorization of the relative entropy in terms of many more conditional relative entropies. Instead of tripartite states, we consider now multipartite ones. To simplify notation, let $\hs_\Lambda= \un{x\in \Lambda}{\bigotimes} \hs_x $ be a multipartite Hilbert space, and let $\rho_\Lambda, \sigma_\Lambda \in \SSS_\Lambda$. We will prove that the relative entropy of both states is upper bounded by the sum of all the conditional relative entropies in every $x \in \Lambda$. The multiplicative error term again disappears, since the state considered here is product.

\begin{theorem}\label{thm:factorization}
Let $\Lambda$ be a finite set and let $\rho_{\Lambda}, \sigma_{\Lambda} \in \SSS_\Lambda$ such that $\ds \sigma_\Lambda = \un{x \in \Lambda}{\bigotimes} \, \sigma_x$. The following inequality holds:
\begin{equation}\label{eq:factorization}
 \ent {\rho_{\Lambda}} {\sigma_{\Lambda}} \leq \un{x \in \Lambda}{\sum} \entA x {\rho_{\Lambda}} {\sigma_{\Lambda}}.
\end{equation}

\end{theorem}

The proof of this theorem is based on the following result:

\begin{lemma}\label{lemma:shearer}
Let $\Lambda$ be a finite set and $\rho_{\Lambda} \in \SSS_\Lambda$. The following inequality holds:
\begin{equation}\label{eq:shearer}
 S({\rho_{\Lambda}}) \geq \un{x\in \Lambda}{\sum} S(x | x^c)_\rho,
\end{equation}
where $S(x | x^c)_\rho$ is the conditional entropy:
\begin{center}
$S(x | x^c)_\rho= S(\rho_{\Lambda}) - S(\rho_{x^c})$.
\end{center}
\end{lemma}

This result constitutes a particular case of the quantum version of Shearer's inequality. It has been proven in several papers, such as \cite{depol-channels} and \cite{junge-palazuelos}, where the proof is based in the \textit{strong subadditivity property} of the von Neumann entropy \cite{strongsub}. 

We can now proceed to the proof of Theorem \ref{thm:factorization}.
\begin{proof}
Let us rewrite equation (\ref{eq:factorization}) as:
\begin{equation}\label{eq:factorization2}
 \ent {\rho_{\Lambda}} {\sigma_{\Lambda}} -  \un{x \in \Lambda}{\sum} \entA x {\rho_{\Lambda}} {\sigma_{\Lambda}} \leq 0 ,
\end{equation}
where $\entA x {\rho_{\Lambda}} {\sigma_{\Lambda}}$ is given by
\begin{center}
$\ds \entA x {\rho_{\Lambda}} {\sigma_{\Lambda}}= \ent {\rho_{\Lambda}} {\sigma_{\Lambda}} - \ent {\rho_{x^c}} {\sigma_{x^c}} $.
\end{center}

Hence, the left-hand side of the previous inequality can be expressed by:
\begin{align*}
& \ds \ent {\rho_{\Lambda}} {\sigma_{\Lambda}} -  \un{x \in \Lambda}{\sum} \entA x {\rho_{\Lambda}} {\sigma_{\Lambda}} =\\
&= \left( 1- \abs{\Lambda} \right) \ent {\rho_{\Lambda}} {\sigma_{\Lambda}} + \un{x \in \Lambda}{\sum} \ent {\rho_{x^c}} {\sigma_{x^c}} \\
&= \left( 1- \abs{\Lambda} \right) \tr[ \rho_\Lambda \left( \log \rho_\Lambda - \log \sigma_\Lambda + \frac{1}{1- \abs{\Lambda}} \un{x \in \Lambda}{\sum}  \log \rho_{x^c}  -  \frac{1}{1- \abs{\Lambda}} \un{x \in \Lambda}{\sum}  \log  \sigma_{x^c}  \right)  ].
\end{align*}

If we now consider only the terms concerning $ \sigma_\Lambda$, using the fact that  $\ds  \sigma_\Lambda = \un{x \in \Lambda}{\bigotimes} \, \sigma_x$ we have:
\begin{align*}
\ds \left( \abs{\Lambda} - 1 \right) \log \sigma_\Lambda -& \un{x \in \Lambda}{\sum} \log \sigma_{x^c}= \\
&= \left( \abs{\Lambda} - 1 \right)  \un{x \in \Lambda}{\sum} \log \sigma_x - \un{x \in \Lambda}{\sum} \, \un{y \neq x}{\sum} \log \sigma_{y} \\
&=  \left( \abs{\Lambda} - 1 \right)  \un{x \in \Lambda}{\sum} \log \sigma_x -  \left( \abs{\Lambda} - 1 \right) \un{x \in \Lambda}{\sum} \log  \sigma_x =0.
\end{align*}

Therefore,
\begin{center}
$ \ds \ent {\rho_{\Lambda}} {\sigma_{\Lambda}} - \un{x \in \Lambda}{\sum} \entA x {\rho_{\Lambda}} {\sigma_{\Lambda}} = \left(1- \abs{\Lambda} \right) \tr[\rho_\Lambda \left( \log \rho_\Lambda  + \frac{1}{1- \abs{\Lambda}} \un{x \in \Lambda}{\sum}  \log \rho_{x^c} \right)]  $,
\end{center}
and, thus, equation (\ref{eq:factorization2})  can be rewritten as
\begin{equation}\label{eq:fact-von-neumann}
\left(  \abs{\Lambda}-1 \right)  S(\rho_\Lambda)   -  \un{x \in \Lambda}{\sum}    S(\rho_{x^c}) \leq 0,
\end{equation}
where we are denoting by $S(\rho_\Lambda)$ the von Neumann entropy of $\rho_\Lambda$.

Finally, recalling that the conditional entropy is defined as
\begin{center}
$S(x | x^c)_\rho= S(\rho_{\Lambda}) - S(\rho_{x^c})$,
\end{center}
expression (\ref{eq:fact-von-neumann}) is equivalent to (\ref{eq:shearer}), finishing thus the proof.
\end{proof}

In the following section we will use this result of quasi-factorization of the kind (\ref{eq:QF-NonOv}) to obtain a log-Sobolev constant for the heat-bath dynamics, when the fixed state of the evolution is product.

When $\sigma_{ABC}$ is not a product state, the situation is a bit more complicated. Now, a term $\xi(\sigma_{ABC})$ should appear, measuring how far $\sigma_{AC}$ is from a product state, as a multiplicative error term. In the following two subsections, we provide two results of quasi-factorization of the relative entropy, one for the conditional relative entropy and another (weaker) one for the conditional relative entropy by expectations. As we have mentioned above, for both results we will need the explicit expressions for conditional relative entropy and conditional relative entropy by expectations, respectively, as we will not be able to obtain them from the properties in the definitions.

\subsection{Quasi-factorization for the conditional relative entropy}\label{subsec:QF-CRE}

In this subsection, we present a quasi-factorization result for the relative entropy in terms of conditional relative entropies. We need to consider some overlap in the regions where we are conditioning the relative entropies of the RHS due to the envisaged applications in quantum many body systems. In virtue of the identification between quantum and classical spin systems mentioned in Subsection \ref{subsec:classical},  this result can be seen as the quantum analogue of Lemma \ref{lemma:clasico}. We will show that this result is equivalent to \cite[Theorem 1]{us}.

\begin{theorem}[Quasi-factorization for CRE]
\label{thm:qfAB}
Let $\hs_{ABC}=\hs_A \otimes \hs_B \otimes \hs_C$ be a tripartite Hilbert space and $\rho_{ABC}, \sigma_{ABC} \in \SSS_{ABC}$.
Then, the following inequality holds
\begin{equation}\label{eq:qfAB}
(1-2\norm{H(\sigma_{AC})}_\infty)\ent{\rho_{ABC}}{\sigma_{ABC}} \le \entA {AB} {\rho_{ABC}} {\sigma_{ABC}} + \entA {BC} {\rho_{ABC}} {\sigma_{ABC}},
\end{equation}
where
\begin{center}
$\ds H(\sigma_{AC})=  \sigma_{A}^{-1/2} \otimes \sigma_{C}^{-1/2}\,  \sigma_{AC} \, \sigma_{A}^{-1/2} \otimes \sigma_{C}^{-1/2} - \identity_{AC}$.
\end{center}

Note that $H(\sigma_{AC})=0$ if $\sigma_{AC}$ is a tensor product between $A$ and $C$.
\end{theorem}

\begin{proof}

Let us recall \cite[Theorem 1]{us} for $\hs_{AC}= \hs_A \otimes \hs_C$ :

\begin{thm}\label{thm:superadditivity}
 For any bipartite states $\rho_{AC},\sigma_{AC}$:
\begin{equation}\label{eq:superadditivity}
(1+2\|H(\sigma_{AC})\|_{\infty})\ent{\rho_{AC}}{\sigma_{AC}}\ge \ent{\rho_A}{\sigma_A} + \ent{\rho_C}{\sigma_C},
\end{equation}
where
\begin{center}
$H(\sigma_{AC}) =  \sigma_A^{-1/2} \otimes \sigma_C^{-1/2} \, \sigma_{AC} \, \sigma_A^{-1/2} \otimes \sigma_C^{-1/2} - \identity_{AC}$,
\end{center}
and $\identity_{AC} $ denotes the identity operator in $\hs_{AC}$.

Note that $H(\sigma_{AC})=0$ if $\sigma_{AC}=\sigma_A\otimes \sigma_C$.
\end{thm} 

It is enough to prove the equivalence between Theorem \ref{thm:superadditivity} and Theorem \ref{thm:qfAB}.

\underline{Th. \ref{thm:qfAB} $\Rightarrow$  Th. \ref{thm:superadditivity} :} Let $\rho_{ABC}, \sigma_{ABC} \in \SSS_{ABC}$. Then, 
\begin{align*}
(1-2\norm{H(\sigma_{AC})}_\infty)\ent{\rho_{ABC}}{\sigma_{ABC}} & \le \entA {AB} {\rho_{ABC}} {\sigma_{ABC}} + \entA {BC} {\rho_{ABC}} {\sigma_{ABC}}\\ & = 2 \ent {\rho_{ABC}} {\sigma_{ABC}}-  \ent {\rho_{C}} {\sigma_{C}} -  \ent {\rho_{A}} {\sigma_{A}}.
\end{align*}

Rewriting this to have something more similar to inequality (\ref{eq:superadditivity}), we have
\begin{center}
 $\ds (1+2\|H(\sigma_{AC})\|_{\infty})\ent{\rho_{ABC}}{\sigma_{ABC}}\ge \ent{\rho_A}{\sigma_A} + \ent{\rho_C}{\sigma_C}$,
\end{center}
so considering a particular case in which the dimension of $\hs_B$ is $1$ (thus, $\hs_{ABC}= \hs_A \otimes \hs_C$), we have inequality (\ref{eq:superadditivity}).

\underline{ Th. \ref{thm:superadditivity} $\Rightarrow$  Th. \ref{thm:qfAB}:} From the monotonicity of the relative entropy, we know that
\begin{center}
$ \ent{\rho_{ABC}}{\sigma_{ABC}} \geq \ent{\rho_{AC}}{\sigma_{AC}} $,
\end{center}
and using this together with inequality (\ref{eq:superadditivity}), we have
\begin{center}
 $\ds (1+2\|H(\sigma_{AC})\|_{\infty})\ent{\rho_{ABC}}{\sigma_{ABC}}\ge \ent{\rho_A}{\sigma_A} + \ent{\rho_C}{\sigma_C}$,
\end{center}
which we have just seen that is a reformulation of inequality (\ref{eq:qfAB}).

\end{proof}

\begin{remark}
It is clear that Proposition \ref{prop:part-qf} constitutes a particular case of this theorem where the multiplicative error term disappears, since in that case we were considering $\sigma_{ABC}$ a tensor product. 

\end{remark}

This result is the strongest one that we present in this paper of the kind (\ref{eq:QF-Ov}). A complete proof can be consulted in \cite{us}. We leave for future work the possibility of using this result for similar purposes than Theorem \ref{thm:factorization}, under a sufficiently strong assumption on the decay of correlations in $\sigma$.

\subsection{Quasi-factorization for the conditional relative entropy by expectations}\label{subsec:QF-CREexp}

Now we move to the definition of conditional relative entropies for expectations. For this definition, we can prove the following result, which is an example of (\ref{eq:QF-NonOv}) for two subregions.

\begin{theorem}[Quasi-factorization for conditional expectations]
\label{thm:quasifactorizationexp}
Let $\hs_{AB}=\hs_A \otimes \hs_B $ be a bipartite Hilbert space and $\rho_{AB}, \sigma_{AB} \in \SSS_{AB}$.
Then, the following inequality holds
\begin{equation}\label{eq:quasifactorizationexp}
(1-\xi(\sigma_{AB}))\ent{\rho_{AB}}{\sigma_{AB}} \le \enteA {A} {\rho_{AB}} {\sigma_{AB}} + \enteA {B} {\rho_{AB}} {\sigma_{AB}},
\end{equation}
where
\begin{center}
$\ds \xi(\sigma_{AB})=  2  \left( E_1(t) +  E_2(t)  \right)     $,
\end{center}
and
\begin{center}
$\ds E_1(t) =  \int_{-\infty}^{+\infty} dt \, \beta_0 (t)  \norm{ \sigma_{B}^{\frac{-1+it}{2}}   \sigma_{AB}^{\frac{1-it}{2}}   \sigma_{A}^{\frac{-1+it}{2}}   -\identity_{AB} }_\infty \norm{  \sigma_{A}^{-1/2}  \sigma_{AB}^{\frac{1+it}{2}}  \sigma_{B}^{-1/2} }_\infty  $, \\
$\ds E_2 (t) = \int_{-\infty}^{+\infty} dt \, \beta_0 (t) \, \norm{ \sigma_{B}^{\frac{-1-it}{2}} \,   \sigma_{AB}^{\frac{1+it}{2}}  \, \sigma_{A}^{\frac{-1-it}{2}}   -\identity_{AB} }_\infty$,
\end{center}
with
 \begin{center}
 $\ds \beta_0(t)= \frac{\pi}{2} (\cosh(\pi t)+ 1)^{-1} . $
 \end{center}
Note that $\xi(\sigma_{AB})=0$ if $\sigma_{AB}$ is a tensor product between $A$ and $B$.
\end{theorem} 

This proof can be split into four steps. The first part of the proof is analogous to the one of \cite[Theorem 1]{us}, but we include it here for the sake of clearness. However, from the second half of the second step, the proof gets much more complicated, leading to the error term shown in the statement of the theorem, which, despite going in the same spirit than its analogue in  Theorem \ref{thm:qfAB}, is less intuitive.

%, but because of the presence of the factor $\sigma_{\text{min}}^{-2}$, which forbids us to use it for the desired purpose, i.e., for proving positivity of a log-Sobolev constant in a quantum spin system (where the system will be taken to the thermodynamic limit and, hence, $\sigma_{\text{min}}^{-2}$ will approach $\infty$ quickly).

\begin{step2}
\label{step:21}
For density matrices $\rho_{AB}, \sigma_{AB} \in \SSS_{AB}$, it holds that
\begin{equation}\label{eq:step-21}
\ent {\rho_{AB}} {\sigma_{AB}} \le \enteA{A} {\rho_{AB}} {\sigma_{AB}} + \enteA{B} {\rho_{AB}} {\sigma_{AB}} + \log \tr M,
\end{equation}
where
$  M = \exp \bqty{ - \log {\sigma_{AB}} + \log \E_{A}^*(\rho_{AB})  + \log \E_{B}^*(\rho_{AB})  } $
and equality holds (both sides being equal to zero) if $\rho_{AB} =
\sigma_{AB}$. \\
Moreover, if $\sigma_{AB} = \sigma_A \otimes \sigma_B$, then $\log \tr M =0$.
\end{step2}

From the definition of conditional relative entropy by expectations it follows that:
\begin{align*}
\ds & \ent {\rho_{AB}} {\sigma_{AB}} - \enteA{A} {\rho_{AB}} {\sigma_{AB}} - \enteA{B} {\rho_{AB}} {\sigma_{AB}}= \\ &= \ent {\rho_{AB}} {\sigma_{AB}}  - \ent {\rho_{AB}}  {\E_{A}^*(\rho_{AB})}  - \ent {\rho_{AB}}  {\E_{B}^*(\rho_{AB})}\\
&=\tr \left[ {\rho_{AB}} \left( - \log {\rho_{AB}} \underbrace{  - \log {\sigma_{AB}}  + \log \E_{A}^*(\rho_{AB})+\log \E_{B}^*(\rho_{AB}) }_{\log M} \right) \right] \\
& = - \ent {\rho_{AB}} M.
\end{align*}

Now, since $\tr[M]\neq 1$ in general,
\begin{center}
$ \ds \ent {\rho_{AB}} M = \ent {\rho_{AB}} {M/\tr[M]} -\log \tr[M] \geq  -\log \tr[M] $,
\end{center}
due to the non-negativity property of the relative entropy.

If $\rho_{AB} = \sigma_{AB}$, $\E_{A}^*(\rho_{AB})= \sigma_{AB}$, and the same for $\E_{B}^*$, so $\log M = \log \sigma_{AB}$ and both sides are equal to zero. Also, if $\sigma_{AB} = \sigma_A \otimes \sigma_B$, we have $ \E_{A}^*(\rho_{AB})= \sigma_A \otimes \rho_B$ and  $ \E_{B}^*(\rho_{AB})= \rho_A \otimes \sigma_B$, so $M=\rho_A \otimes \rho_B$. Hence, $\log \tr M= 0$. 

\begin{step2}\label{step:22}
With the same notation of step \ref{step:21}, we have that
\begin{equation}
\log \tr M \le \int_{- \infty}^{+ \infty} dt \beta_0(t) \tr[\sigma_A^{-1/2} (\rho_A - \sigma_A) \sigma_A^{-1/2} \sigma_{AB}^{\frac{1-it}{2}} \sigma_B^{-1/2} (\rho_B - \sigma_B) \sigma_B^{-1/2} \sigma_{AB}^{\frac{1+it}{2}} ],
\end{equation}
with
 \begin{center}
 $\ds \beta_0(t)= \frac{\pi}{2} (\cosh(\pi t)+ 1)^{-1}  $.
 \end{center}

\end{step2}

\begin{proof}

In the proof of this step, we make use of the following theorem of Lieb \cite{lieb}, which extends the Golden-Thompson inequality.

\begin{thm}
\label{thm:Lieb}
Let $f, g$ be Hermitian operators, and define
\begin{equation}
	\mathcal T_g(f) = \int_0^\infty \dd{t} (g+t)^{-1} f (g+t)^{-1} .
\end{equation}
$\mathcal T_g$ is positive-semidefinite if $g$ is. 

For $h$ a Hermitian operator, we have that
\begin{equation}
	\tr[ \exp(-f+g+h)] \le \tr[ e^h \mathcal T_{e^f}(e^g)].
\end{equation}
\end{thm}

We use an alternative definition of this superoperator to obtain a necessary tool for the proof of Step \ref{step:22}. In \cite[Lemma 3.4]{sutter}, Sutter, Berta and Tomamichel prove the following result:

\begin{lemma}\label{lemma:lieb-sutter}
For  ${f} $ a positive semidefinite operator and $g $ a Hermitian operator the following holds:
\begin{center}
$\ds \mathcal{T}_g(f)= \int_{-\infty}^{+\infty} dt \, \beta_0 (t) \, g^{\frac{-1-it}{2}} \,f \, g^{\frac{-1+it}{2}} ,$
\end{center}
with
 \begin{center}
 $\ds \beta_0(t)= \frac{\pi}{2} (\cosh(\pi t)+ 1)^{-1}  $.
 \end{center}
\end{lemma} 

Applying Theorem \ref{thm:Lieb} to inequality (\ref{eq:step-21}), we have
\begin{align*}
\tr M =&
\tr\left[ \exp( - \underbrace{\log \sigma_{AB} }_{f} +
\underbrace{\log \E_{A}^*(\rho_{AB})}_{h} +
\underbrace{\log\E_{B}^*(\rho_{AB}) }_{g} ) \right] \\
\leq & \tr[ \E_{A}^*(\rho_{AB})\mathcal T_{\sigma_{AB}} (\E_{B}^*(\rho_{AB}))],
\end{align*}
and in virtue of Lemma \ref{lemma:lieb-sutter}, 
\begin{center}
$\ds \tr M \leq \tr[ \E_{A}^*(\rho_{AB})  \int_{-\infty}^{+\infty} dt \, \beta_0 (t) \, \sigma_{AB}^{\frac{-1-it}{2}}  \E_{B}^*(\rho_{AB})  \sigma_{AB}^{\frac{-1+it}{2}}] $.
\end{center}

Now, replacing the values of $ \E_{A}^*(\rho_{AB})  $ and $\E_{B}^*(\rho_{AB}) $, and using the linearity of the trace, we have
\begin{align*}
\tr M & \leq  \int_{-\infty}^{+\infty} dt \, \beta_0 (t) \,\tr[ \E_{A}^*(\rho_{AB})   \sigma_{AB}^{\frac{-1-it}{2}}  \E_{B}^*(\rho_{AB})  \sigma_{AB}^{\frac{-1+it}{2}}] \\
& =\int_{-\infty}^{+\infty} dt \, \beta_0 (t) \tr[  \sigma_{AB}^{1/2}  \sigma_{B}^{-1/2} \rho_B \sigma_{B}^{-1/2} \sigma_{AB}^{1/2}   \sigma_{AB}^{\frac{-1-it}{2}} \sigma_{AB}^{1/2}  \sigma_{A}^{-1/2} \rho_A \sigma_{A}^{-1/2} \sigma_{AB}^{1/2}  \sigma_{AB}^{\frac{-1+it}{2}}   ] \\
& =\int_{-\infty}^{+\infty} dt \, \beta_0 (t) \tr[  \sigma_{B}^{-1/2} \rho_B \sigma_{B}^{-1/2}   \sigma_{AB}^{\frac{1-it}{2}}   \sigma_{A}^{-1/2} \rho_A \sigma_{A}^{-1/2}  \sigma_{AB}^{\frac{1+it}{2}}   ]
\end{align*}

If we substract $\sigma_B$ from $\rho_B$ and $\sigma_A$ from $\rho_A$ in the term inside the integral of the previous expression, we have
\begin{align*}
& \tr[  \sigma_{B}^{-1/2} \left( \rho_B  - \sigma_B \right) \sigma_{B}^{-1/2}   \sigma_{AB}^{\frac{1-it}{2}}   \sigma_{A}^{-1/2} \left( \rho_A - \sigma_A \right)  \sigma_{A}^{-1/2}  \sigma_{AB}^{\frac{1+it}{2}}   ] = \\
& \phantom{asdASadsa} = \tr[  \sigma_{B}^{-1/2} \rho_B \sigma_{B}^{-1/2}   \sigma_{AB}^{\frac{1-it}{2}}   \sigma_{A}^{-1/2} \rho_A \sigma_{A}^{-1/2}  \sigma_{AB}^{\frac{1+it}{2}}   ] \\
& \phantom{asdadsASaaaa} + \tr[  \sigma_{B}^{-1/2} \sigma_B \sigma_{B}^{-1/2}   \sigma_{AB}^{\frac{1-it}{2}}   \sigma_{A}^{-1/2} \sigma_A \sigma_{A}^{-1/2}  \sigma_{AB}^{\frac{1+it}{2}}   ] \\
& \phantom{asdaASdsaaaa} - \tr[  \sigma_{B}^{-1/2} \sigma_B \sigma_{B}^{-1/2}   \sigma_{AB}^{\frac{1-it}{2}}   \sigma_{A}^{-1/2} \rho_A \sigma_{A}^{-1/2}  \sigma_{AB}^{\frac{1+it}{2}}   ] \\
& \phantom{asdaASdsaaaa} - \tr[  \sigma_{B}^{-1/2} \rho_B \sigma_{B}^{-1/2}   \sigma_{AB}^{\frac{1-it}{2}}   \sigma_{A}^{-1/2} \sigma_A \sigma_{A}^{-1/2}  \sigma_{AB}^{\frac{1+it}{2}}   ] \\
& \phantom{asdASadsa} = \left( \tr[  \sigma_{B}^{-1/2} \rho_B \sigma_{B}^{-1/2}   \sigma_{AB}^{\frac{1-it}{2}}   \sigma_{A}^{-1/2} \rho_A \sigma_{A}^{-1/2}  \sigma_{AB}^{\frac{1+it}{2}}   ] + \tr[  \sigma_{AB}^{\frac{1-it}{2}}   \sigma_{AB}^{\frac{1+it}{2}}   ]  \right) \\
& \phantom{asdaASdsaaaa} - \left( \tr[  \sigma_{AB}^{\frac{1-it}{2}}   \sigma_{A}^{-1/2} \rho_A \sigma_{A}^{-1/2}  \sigma_{AB}^{\frac{1+it}{2}}   ] + \tr[  \sigma_{B}^{-1/2} \rho_B \sigma_{B}^{-1/2}   \sigma_{AB}^{\frac{1-it}{2}}  \sigma_{AB}^{\frac{1+it}{2}}   ] \right) \\
& \phantom{asdaASdsa} = \left( \tr[  \sigma_{B}^{-1/2} \rho_B \sigma_{B}^{-1/2}   \sigma_{AB}^{\frac{1-it}{2}}   \sigma_{A}^{-1/2} \rho_A \sigma_{A}^{-1/2}  \sigma_{AB}^{\frac{1+it}{2}}   ] + 1 \right) \\
& \phantom{asdadASsaaaa} - \left( \tr[  \sigma_{A} \sigma_{A}^{-1/2} \rho_A \sigma_{A}^{-1/2} ] + \tr[  \sigma_{B}^{-1/2} \rho_B \sigma_{B}^{-1/2}  \sigma_{B} ] \right) \\
& \phantom{asdASadsa} = \tr[  \sigma_{B}^{-1/2} \rho_B \sigma_{B}^{-1/2}   \sigma_{AB}^{\frac{1-it}{2}}   \sigma_{A}^{-1/2} \rho_A \sigma_{A}^{-1/2}  \sigma_{AB}^{\frac{1+it}{2}}   ] + 1 -1 -1,
\end{align*}
where we have used some properties of the trace, such as its linearity, cyclicity and the fact that if $f_A \in \A_A$ and $g_{AB} \in \SSS_{AB} $  then
\begin{center}
$\ds \tr[f_A g_{AB}] = \tr[f_A g_A]  $.
\end{center}

Therefore, we have the following equality:
\begin{align*}
&\tr[  \sigma_{B}^{-1/2} \rho_B \sigma_{B}^{-1/2}   \sigma_{AB}^{\frac{1-it}{2}}   \sigma_{A}^{-1/2} \rho_A \sigma_{A}^{-1/2}  \sigma_{AB}^{\frac{1+it}{2}}   ]   = \\
& \phantom{asdadsa} =  \tr[  \sigma_{B}^{-1/2} \left( \rho_B  - \sigma_B \right) \sigma_{B}^{-1/2}   \sigma_{AB}^{\frac{1-it}{2}}   \sigma_{A}^{-1/2} \left( \rho_A - \sigma_A \right)  \sigma_{A}^{-1/2}  \sigma_{AB}^{\frac{1+it}{2}}   ]  + 1,
\end{align*}
and hence
\begin{center}
$\ds \tr M \leq \int_{-\infty}^{+\infty} dt \, \beta_0 (t) \left(  \tr[  \sigma_{B}^{-1/2} \left( \rho_B  - \sigma_B \right) \sigma_{B}^{-1/2}   \sigma_{AB}^{\frac{1-it}{2}}   \sigma_{A}^{-1/2} \left( \rho_A - \sigma_A \right)  \sigma_{A}^{-1/2}  \sigma_{AB}^{\frac{1+it}{2}}   ]  + 1 \right) $
\end{center}

If we now use the following inequality for positive real numbers
\begin{center}
 $\log(x)\le x-1$,
\end{center}
and the monotonicity of the logarithm and the fact that $\beta_0(t)$ integrates $1$, we can then conclude
\begin{align*}
& \log \tr M \leq \\
& \; \leq \log \left[ \int_{-\infty}^{+\infty} dt \, \beta_0 (t) \left(  \tr[  \sigma_{B}^{-1/2} \left( \rho_B  - \sigma_B \right) \sigma_{B}^{-1/2}   \sigma_{AB}^{\frac{1-it}{2}}   \sigma_{A}^{-1/2} \left( \rho_A - \sigma_A \right)  \sigma_{A}^{-1/2}  \sigma_{AB}^{\frac{1+it}{2}}   ]  + 1 \right) \right] \\
& \; = \log \left[ \int_{-\infty}^{+\infty} dt \, \beta_0 (t)   \tr[  \sigma_{B}^{-1/2} \left( \rho_B  - \sigma_B \right) \sigma_{B}^{-1/2}   \sigma_{AB}^{\frac{1-it}{2}}   \sigma_{A}^{-1/2} \left( \rho_A - \sigma_A \right)  \sigma_{A}^{-1/2}  \sigma_{AB}^{\frac{1+it}{2}}   ]  + 1 \right]\\
& \; \leq \int_{-\infty}^{+\infty} dt \, \beta_0 (t)  \tr[  \sigma_{B}^{-1/2} \left( \rho_B  - \sigma_B \right) \sigma_{B}^{-1/2}   \sigma_{AB}^{\frac{1-it}{2}}   \sigma_{A}^{-1/2} \left( \rho_A - \sigma_A \right)  \sigma_{A}^{-1/2}  \sigma_{AB}^{\frac{1+it}{2}}   ].  
\end{align*}

\end{proof}

\begin{step2}\label{step:23}
With the same notation of the previous steps,
\begin{equation}
 \tr[  \sigma_{B}^{-1/2} \left( \rho_B  - \sigma_B \right) \sigma_{B}^{-1/2}   \sigma_{AB}^{\frac{1-it}{2}}   \sigma_{A}^{-1/2} \left( \rho_A - \sigma_A \right)  \sigma_{A}^{-1/2}  \sigma_{AB}^{\frac{1+it}{2}}   ] = \xi_1 + \xi_2,
\end{equation}
where
\begin{center}
$\ds \xi_1 = \tr[ T_B \left( \sigma_{AB}^{\frac{1-it}{2}}-\left( \sigma_A \otimes \sigma_B  \right)^{\frac{1-it}{2}} \right) T_A \, \sigma_{AB}^{\frac{1+it}{2}}  ]$, \\
$\ds \xi_2 =\tr[ T_B  \left( \sigma_A \otimes \sigma_B  \right)^{\frac{1-it}{2}}  T_A  \left( \sigma_{AB}^{\frac{1+it}{2}} -\left( \sigma_A \otimes \sigma_B   \right)^{\frac{1+it}{2}}  \right) ] $,
\end{center}
for certain observables $T_A \in \A_A$ and $T_B \in \A_B$.

Notice that both $\xi_1$ and $\xi_2$ vanish when $\sigma_{AB}$ is a tensor product.

\end{step2}

\begin{proof}

Let us first denote
\begin{center}
$ T_A :=  \sigma_{A}^{-1/2} \left( \rho_A - \sigma_A \right)  \sigma_{A}^{-1/2} $,\\
$ T_B :=  \sigma_{B}^{-1/2} \left( \rho_B  - \sigma_B \right) \sigma_{B}^{-1/2}   $,
\end{center}
to simplify notation. Hence
\begin{center}
$\ds \tr[  \sigma_{B}^{-1/2} \left( \rho_B  - \sigma_B \right) \sigma_{B}^{-1/2}   \sigma_{AB}^{\frac{1-it}{2}}   \sigma_{A}^{-1/2} \left( \rho_A - \sigma_A \right)  \sigma_{A}^{-1/2}  \sigma_{AB}^{\frac{1+it}{2}}   ] = \tr[ T_B \sigma_{AB}^{\frac{1-it}{2}} T_A   \sigma_{AB}^{\frac{1+it}{2}}  ] $
\end{center}

Now, we add and substract $\left( \sigma_A \otimes \sigma_B  \right)^{\frac{1-it}{2}} $ to $  \sigma_{AB}^{\frac{1-it}{2}}  $ and $\left( \sigma_A \otimes \sigma_B  \right)^{\frac{1+it}{2}} $ to $  \sigma_{AB}^{\frac{1+it}{2}}  $. We will later use some combinations of these terms in the error terms, so that we recover the property that the error terms vanish whenever $\sigma_{AB}$ is a tensor product. 
\begin{align*}
 \tr[ T_B \sigma_{AB}^{\frac{1-it}{2}} T_A   \sigma_{AB}^{\frac{1+it}{2}}  ] &=  \tr \left[ T_B \left( \sigma_{AB}^{\frac{1-it}{2}}-\left( \sigma_A \otimes \sigma_B   \right)^{\frac{1-it}{2}} +\left( \sigma_A \otimes \sigma_B  \right)^{\frac{1-it}{2}}  \right)  \right. \\
& \phantom{ass} \left. \cdot \, T_A  \left( \sigma_{AB}^{\frac{1+it}{2}} -\left( \sigma_A \otimes \sigma_B   \right)^{\frac{1+it}{2}} +\left( \sigma_A \otimes \sigma_B   \right)^{\frac{1+it}{2}}  \right) \right] \\
& = \underbrace{\tr[ T_B \left( \sigma_{AB}^{\frac{1-it}{2}}-\left( \sigma_A \otimes \sigma_B  \right)^{\frac{1-it}{2}} \right) T_A  \sigma_{AB}^{\frac{1+it}{2}}  ]}_{\xi_1} \\
& \phantom{ass} + \underbrace{\tr[ T_B  \left( \sigma_A \otimes \sigma_B  \right)^{\frac{1-it}{2}}  T_A  \left( \sigma_{AB}^{\frac{1+it}{2}} -\left( \sigma_A \otimes \sigma_B   \right)^{\frac{1+it}{2}}  \right) ]}_{\xi_2} \\
& \phantom{ass} +\underbrace{\tr[ T_B  \left( \sigma_A \otimes \sigma_B \right)^{\frac{1-it}{2}}  T_A  \left( \sigma_A \otimes \sigma_B \right)^{\frac{1+it}{2}}   ]}_{\xi_3}.
\end{align*}

There is only left to prove that $\xi_3$ is $0$. For that, let us replace again the values of $T_A$ and $T_B$ in the expression of $\xi_3$. 
\begin{align*}
\xi_3 & = \tr[ \sigma_{B}^{-1/2} \left( \rho_B  - \sigma_B \right) \sigma_{B}^{-1/2}  \left( \sigma_A \otimes \sigma_B \right)^{\frac{1-it}{2}}  \sigma_{A}^{-1/2} \left( \rho_A - \sigma_A \right)  \sigma_{A}^{-1/2} \left( \sigma_A \otimes \sigma_B \right)^{\frac{1+it}{2}}   ]\\
&= \tr[  \sigma_B^{\frac{1+it}{2}} \sigma_{B}^{-1/2} \left( \rho_B  - \sigma_B \right) \sigma_{B}^{-1/2} \sigma_B^{\frac{1-it}{2}}  \sigma_A^{\frac{1-it}{2}}   \sigma_{A}^{-1/2} \left( \rho_A - \sigma_A \right)  \sigma_{A}^{-1/2} \sigma_A^{\frac{1+it}{2}}    ] \\
&= \tr[ \sigma_B^{\frac{1+it}{2}} \sigma_{B}^{-1/2} \left( \rho_B - \sigma_B \right) \sigma_{B}^{-1/2} \sigma_B^{\frac{1-it}{2}} ] \tr[\sigma_A^{\frac{1-it}{2}}   \sigma_{A}^{-1/2} \left( \rho_A - \sigma_A \right)  \sigma_{A}^{-1/2} \sigma_A^{\frac{1+it}{2}}  ] \\
&= \tr[\rho_B - \sigma_B] \tr[\rho_A - \sigma_A]\\
&= 0,
\end{align*}
where we have used the fact that states with disjoint supports commute and the factorization of the trace under tensor products.

Therefore, 
\begin{center}
$\ds \tr[ T_B \, \sigma_{AB}^{\frac{1-it}{2}} \, T_A  \, \sigma_{AB}^{\frac{1+it}{2}}  ] = \xi_1 + \xi_2$.
\end{center}
\end{proof}
\begin{step2}\label{step:24}
With the same notation of the previous steps:
\begin{center}
$\ds \log \tr M \leq 2    \left( \int_{-\infty}^{+\infty} dt \, \beta_0 (t) \, \norm{  S_1(t)}_\infty \norm{  \sigma_{A}^{-1/2}  \sigma_{AB}^{\frac{1+it}{2}}  \sigma_{B}^{-1/2} }_\infty + \int_{-\infty}^{+\infty} dt \, \beta_0 (t) \, \norm{  S_2(t)}_\infty  \right) \ent {\rho_{AB}} {\sigma_{AB}} $,
\end{center}
where $S_1(t)$ and $S_2(t)$ depend only on $\sigma_{AB}$ and vanish when $\sigma_{AB}= \sigma_A \otimes \sigma_B$.
\end{step2}

\begin{proof} Let us bound separately $\xi_1$ and $\xi_2$.

First, we denote:
\begin{center}
$\ds S_1(t) :=  \sigma_{AB}^{\frac{1-it}{2}}-\left( \sigma_A \otimes \sigma_B  \right)^{\frac{1-it}{2}} $, \\
$\ds S_2 (t) := \sigma_{AB}^{\frac{1+it}{2}} -\left( \sigma_A \otimes \sigma_B  \right)^{\frac{1+it}{2}} $,
\end{center}
again to simplify notation. Using the submultiplicativity of the Schatten norms, we have for $\xi_1$
\begin{align*}
\xi_1 &=  \int_{-\infty}^{+\infty} dt \, \beta_0 (t) \, \tr[ T_B \left( \sigma_{AB}^{\frac{1-it}{2}}-\left( \sigma_A \otimes \sigma_B   \right)^{\frac{1-it}{2}} \right) T_A  \sigma_{AB}^{\frac{1+it}{2}}  ] \\
&= \int_{-\infty}^{+\infty} dt \, \beta_0 (t) \, \tr[ \left( \rho_B  - \sigma_B \right) \sigma_{B}^{-1/2} \,  S_1(t)\, \sigma_{A}^{-1/2} \left( \rho_A - \sigma_A \right)  \sigma_{A}^{-1/2}  \sigma_{AB}^{\frac{1+it}{2}}  \sigma_{B}^{-1/2} ]\\
& \leq \norm{\rho_B - \sigma_B}_1\int_{-\infty}^{+\infty} dt \, \beta_0 (t) \, \norm{ \sigma_{B}^{-1/2} \,  S_1(t)\, \sigma_{A}^{-1/2} \left( \rho_A - \sigma_A \right)  \sigma_{A}^{-1/2}  \sigma_{AB}^{\frac{1+it}{2}}  \sigma_{B}^{-1/2} }_1
\end{align*}
and in virtue of Hölder's inequality,
\begin{align*}
\xi_1 & \leq \norm{\rho_B- \sigma_B}_1 \int_{-\infty}^{+\infty} dt \, \beta_0 (t) \, \norm{ \sigma_{B}^{-1/2} \,  S_1(t)\, \sigma_{A}^{-1/2} \left( \rho_A - \sigma_A \right)  \sigma_{A}^{-1/2}  \sigma_{AB}^{\frac{1+it}{2}}  \sigma_{B}^{-1/2} }_1 \\
& \leq \norm{\rho_B - \sigma_B}_1 \int_{-\infty}^{+\infty} dt \, \beta_0 (t) \, \norm{ \sigma_{B}^{-1/2} \,  S_1(t)\, \sigma_{A}^{-1/2} }_\infty   \norm{\left( \rho_A - \sigma_A \right)  \sigma_{A}^{-1/2}  \sigma_{AB}^{\frac{1+it}{2}}  \sigma_{B}^{-1/2} }_1 \\
& \leq \norm{\rho_B - \sigma_B}_1  \norm{\rho_A - \sigma_A }_1  \int_{-\infty}^{+\infty} dt \, \beta_0 (t) \, \norm{ \sigma_{B}^{-1/2} \,  S_1(t)\, \sigma_{A}^{-1/2} }_\infty   \norm{  \sigma_{A}^{-1/2}  \sigma_{AB}^{\frac{1+it}{2}}  \sigma_{B}^{-1/2} }_\infty.
\end{align*}

Now, for the first norm inside the integral, we have 
\begin{align*}
 \norm{ \sigma_{B}^{-1/2} \,  S_1(t)\, \sigma_{A}^{-1/2} }_\infty &=  \norm{ \sigma_{B}^{-1/2} \,  \left( \sigma_{AB}^{\frac{1-it}{2}}-\left( \sigma_A \otimes \sigma_B  \right)^{\frac{1-it}{2}} \right)\, \sigma_{A}^{-1/2} }_\infty \\
 &= \norm{ \sigma_{B}^{-1/2} \,   \sigma_{AB}^{\frac{1-it}{2}}  \, \sigma_{A}^{-1/2}   -\left( \sigma_A \otimes \sigma_B  \right)^{\frac{-it}{2}}  }_\infty \\
  &= \norm{ \sigma_{B}^{\frac{-1+it}{2}} \,   \sigma_{AB}^{\frac{1-it}{2}}  \, \sigma_{A}^{\frac{-1+it}{2}}   -\identity_{AB} }_\infty,
\end{align*}
because of the unitarily invariance of Schatten norms.

Finally, using Pinsker's inequality and the data-processing inequality, we have:
\begin{center}
$\ds  \norm{\rho_B - \sigma_B}_1 \leq \sqrt{2 \ent {\rho_B} {\sigma_{B}} }\leq  \sqrt{ 2\ent {\rho_{AB}} {\sigma_{AB}} } $,
\end{center}
and analogously for the term with support in $A$. Thus, we can bound $\xi_1$ by
\begin{center}
$\ds  \xi_1 \leq \left( 2  \int_{-\infty}^{+\infty} dt \, \beta_0 (t)  \norm{ \sigma_{B}^{\frac{-1+it}{2}}   \sigma_{AB}^{\frac{1-it}{2}}   \sigma_{A}^{\frac{-1+it}{2}}   -\identity_{AB} }_\infty \norm{  \sigma_{A}^{-1/2}  \sigma_{AB}^{\frac{1+it}{2}}  \sigma_{B}^{-1/2} }_\infty  \right)  \ent {\rho_{AB}} {\sigma_{AB}} $.
\end{center}

We can do the same for $\xi_2$. First,
\begin{align*}
\xi_2 &=  \int_{-\infty}^{+\infty} dt \, \beta_0 (t) \, \tr[ T_B  \left( \sigma_A \otimes \sigma_B  \right)^{\frac{1-it}{2}}  T_A  \left( \sigma_{AB}^{\frac{1+it}{2}} -\left( \sigma_A \otimes \sigma_B  \right)^{\frac{1+it}{2}}  \right) ]  \\
&=  \int_{-\infty}^{+\infty} dt \, \beta_0 (t) \, \tr[  \left( \rho_B  - \sigma_B \right) \sigma_{B}^{-1/2}  \left( \sigma_A \otimes \sigma_B  \right)^{\frac{1-it}{2}} \sigma_{A}^{-1/2} \left( \rho_A - \sigma_A \right)  \sigma_{A}^{-1/2} S_2(t)  \,  \sigma_{B}^{-1/2} ]  \\
& \leq \norm{\rho_B - \sigma_B}_1\int_{-\infty}^{+\infty} dt \, \beta_0 (t) \, \norm{ \sigma_{B}^{-1/2}  \left( \sigma_A \otimes \sigma_B  \right)^{\frac{1-it}{2}} \sigma_{A}^{-1/2} \left( \rho_A - \sigma_A \right)  \sigma_{A}^{-1/2} S_2(t)  \,  \sigma_{B}^{-1/2}}_1,
\end{align*}
where we have used the submultiplicativity of Schatten norms. Using again Hölder's inequality twice, we can bound this term by:
\begin{align*}
\xi_2 &\leq \norm{\rho_B - \sigma_B}_1\int_{-\infty}^{+\infty} dt \, \beta_0 (t) \, \norm{ \sigma_{B}^{-1/2}  \left( \sigma_A \otimes \sigma_B   \right)^{\frac{1-it}{2}} \sigma_{A}^{-1/2} \left( \rho_A - \sigma_A \right)  \sigma_{A}^{-1/2} S_2(t) \,  \sigma_{B}^{-1/2}}_1 \\
& \leq \norm{\rho_B - \sigma_B}_1\int_{-\infty}^{+\infty} dt \, \beta_0 (t) \, \norm{ \sigma_{B}^{-1/2}  \left( \sigma_A \otimes \sigma_B  \right)^{\frac{1-it}{2}} \sigma_{A}^{-1/2}}_\infty \norm{ \left( \rho_A - \sigma_A \right)  \sigma_{A}^{-1/2} S_2(t) \,  \sigma_{B}^{-1/2}}_1 \\
& \leq \norm{\rho_B - \sigma_B}_1  \norm{\rho_A - \sigma_A}_1 \int_{-\infty}^{+\infty} dt \, \beta_0 (t) \, \norm{ \sigma_{B}^{-1/2}  \left( \sigma_A \otimes \sigma_B   \right)^{\frac{1-it}{2}} \sigma_{A}^{-1/2}}_\infty \norm{  \sigma_{A}^{-1/2} S_2(t) \,  \sigma_{B}^{-1/2}}_\infty
\end{align*}

For the first term inside the integral, it is clear that 
\begin{center}
$\ds  \norm{ \sigma_{B}^{-1/2}  \left( \sigma_A \otimes \sigma_B \right)^{\frac{1-it}{2}} \sigma_{A}^{-1/2}}_\infty = 1$.
\end{center}

Therefore,
\begin{center}
$\ds \xi_2 \leq   \norm{\rho_B - \sigma_B}_1  \norm{\rho_A - \sigma_A}_1 \,  \int_{-\infty}^{+\infty} dt \, \beta_0 (t) \, \norm{  \sigma_{A}^{-1/2} S_2(t) \,  \sigma_{B}^{-1/2}}_\infty$,
\end{center}
and again
\begin{center}
$\ds   \norm{  \sigma_{A}^{-1/2} S_2(t) \,  \sigma_{B}^{-1/2}}_\infty=   \norm{ \sigma_{B}^{\frac{-1-it}{2}} \,   \sigma_{AB}^{\frac{1+it}{2}}  \, \sigma_{A}^{\frac{-1-it}{2}}   -\identity_{AB} }_\infty$
\end{center}
because of the unitarily invariance of Schatten norms.

Finally, as in the case of $\xi_1$, in virtue of Pinsker's inequality and the data-processing inequality, we obtain:
\begin{center}
$\ds  \xi_2 \leq \left( 2  \,  \int_{-\infty}^{+\infty} dt \, \beta_0 (t) \, \norm{ \sigma_{B}^{\frac{-1-it}{2}} \,   \sigma_{AB}^{\frac{1+it}{2}}  \, \sigma_{A}^{\frac{-1-it}{2}}   -\identity_{AB} }_\infty  \right) \ent {\rho_{AB}} {\sigma_{AB}} $
\end{center}

Notice that when $\sigma_{AB}= \sigma_{A} \otimes \sigma_B $, both $S_1(t)$ and $S_2(t)$ vanish, obtaining then an error term for the quasi-factorization result that vanishes when $\sigma_{AB}$ is a product.

\end{proof}

 \begin{remark}
 The result of quasi-factorization obtained in this subsection is much worse than the one obtained in the previous subsection for the conditional relative entropy. This fact might support the idea that the best definition for conditional relative entropy is the original one, and not the modification obtained when taking out the property of semi-monotonicity. 
 
 \end{remark}
 
\begin{remark}
The bounds are not tight. In particular, in the fourth step, we bound $\xi_1$ and $\xi_2$ in a very loose way, giving space to possible improvements of the bounds, and, hence, to a possibly better result of quasi-factorization.
\end{remark} 

\begin{remark}
Similarly to what we mentioned in the previous subsection, Proposition \ref{prop:part-qf2} can be also seen as a particular case of this theorem, when the number of subregions considered is $2$. We notice again that in the simplification given by the proposition the multiplicative error term disappears, since in that case we were considering $\sigma_{ABC}$ a tensor product. 

\end{remark}

\begin{remark}
Throughout the proof of the theorem, we are not using too strongly the fact that we are working with a specific conditional expectaction, the minimal conditional expectaction. The application of Lieb's Theorem of course is independent of this fact, but the bound that follows is not. Going back to the beginning of Step \ref{step:24}, one possible way of defining a more general condition might be the following: From the properties of the conditional expectation, we have
\begin{align*}
\tr \left[ \E_{B}^*(\rho_{AB}-\sigma_{AB}) \right. &\mathcal T_{\sigma_{AB}}  \left.(\E_{A}^*(\rho_{AB}-\sigma_{AB})) \right]= \\
 &=  \tr[ (\E_{B}^*(\rho_{AB}) -\sigma_{AB}) \mathcal T_{\sigma_{AB}}( \E_{A}^*(\rho_{AB}) -\sigma_{AB}
) ] \\
&= \tr[ \E_{B}^*(\rho_{AB}) \mathcal T_{\sigma_{AB}} (\E_{A}^*(\rho_{AB}))]
- \tr[ \E_{A}^*(\rho_{AB})] - \tr[ \E_{B}^*(\rho_{AB})] + \tr[ \sigma_{AB} ]\\
&=\tr[ \E_{B}^*(\rho_{AB}) \mathcal T_{\sigma_{AB}} (\E_{A}^*(\rho_{AB}))] -1,
\end{align*}
where we have used that $\mathcal{T}_{\sigma_{AB}}$ is self-adjoint with respect to the Hilbert-Schmidt product and $\mathcal{T}_{\sigma_{AB}}(\sigma_{AB})= \identity$. Furthermore, we can also write this term as:
\begin{center}
$ \tr[ \E_{B}^*(\rho_{AB}-\sigma_{AB}) \mathcal T_{\sigma_{AB}}
(\E_{A}^*(\rho_{AB}-\sigma_{AB}))]=$\\
$ = \tr[ \E_{B}^*(\rho_{A}-\sigma_{A}) \mathcal T_{\sigma_{AB}}
(\E_{A}^*(\rho_{B}-\sigma_{B}))]$,
\end{center}
since $\E_{A}^*(\eta_{AB})= \E_{A}^*(\eta_{B})$ for every $\eta_{AB} \in \SSS_{AB}$ and the same holds for $\E_{B}^*$.
Therefore, we can directly derive that
\[ \log \tr M \le
\tr[ \E_{B}^*(\rho_{A}-\sigma_{A}) \mathcal T_{\sigma_{AB}}
(\E_{A}^*(\rho_{B}-\sigma_{B}))], \]
for any conditional expectation.
Now let
\[ H = \E_{B} \circ \mathcal T_{\sigma_{AB}} \circ \E_{A}^*, \]
so that $\log \tr M \le \tr[ (\rho_{A}-\sigma_{A}) H(\rho_{B}-\sigma_{B})]$.
Since we have that
\[ \tr[ (\rho_{A} - \sigma_{A}) \identity (\rho_{B} - \sigma_{B}) ]=
\tr[ \rho_{A} - \sigma_{A}] \tr[ \rho_{B}-\sigma_{B}  ]= 0, \]
we can subtract the identity superoperator from the previous bound,
and we obtain that the error term is bounded as follows
\[ \log \tr M \le \norm{ H-\identity }_{\infty} \norm{\rho_{B}-\sigma_{B}}_1\norm{\rho_{A}-\sigma_{A}}_1 ,\]
obtaining a result which is analogous to Steps \ref{step:23} and \ref{step:24}, which were devoted to bounding $\norm{ H-\identity }_{\infty} $ in an appropriate way.
\end{remark}

\section{Quantum spin lattices: log-Sobolev constant}\label{sec:logSob}

In this section, we will study open quantum many body systems, which are weakly coupled to an environment. They constitute realistic physical systems and are relevant for quantum information processing. These systems interact with the environment in a considerable way and, thus, the resulting dynamics are dissipative. We shall use for such systems the Markov approximation, which states that the continuous time evolution of a state of such system is given by a quantum Markov semigroup.

Consider a quantum spin lattice system, which will be assumed to live on a $d$-dimensional finite square lattice, and will be denoted by $\Lambda \subseteq \Z^d$. To every site $x $ in $\Lambda$, we associate a finite dimensional local Hilbert space  $\HH_x$. Then, the Hilbert space associated to the spin lattice $\Lambda$ is given by $\HH_\Lambda= \underset{x \in \Lambda}{\bigotimes} \HH_x $. We denote the set of observables in $\Lambda$ by $\A_\Lambda$, and the set of states by $\SSS_\Lambda$.

In virtue of the Markov approximation mentioned above, in the Schrödinger picture, given an initial state of the system $\rho_{\Lambda} \in \SSS_\Lambda$, its evolution under the dissipative dynamics is given by a quantum Markov semigroup, which is nothing but  a one-parameter family of linear, CPTP maps (quantum channels, \cite{wolf})  $  \qty{\mathcal{T}^*_t}_{t\geq 0} $ on $ \SSS_\Lambda$, verifying:
\begin{enumerate}
\item $\mathcal{T}^*_0= \identity$.
\item $\mathcal{T}^*_t \circ \mathcal{T}^*_s= \mathcal{T}^*_{t+s}$.
\end{enumerate}

The generator of this semigroup is denoted by $\mathcal{L}_\Lambda^*$, called \textit{Lindbladian} (or \textit{Liouvillian}) and satisfies
\begin{equation}\label{eq:gen-semigroup}
\frac{d}{dt} \mathcal{T}^*_t =\mathcal{L}_\Lambda^* \circ \mathcal{T}^*_t = \mathcal{T}^*_t \circ \mathcal{L}_\Lambda^* .
\end{equation}

Thus, we can write the elements of the quantum Markov semigroup as
\begin{center}
$ \mathcal{T}^*_t  = e^{t \mathcal{L}_\Lambda^*} $. 
\end{center}

The notation $^*$ appears since we are in the Schrödinger picture, and denotes that this quantum channel may be seen as the dual of another one in the Heisenberg  picture. Given $\rho_\Lambda \in \SSS_\Lambda$, let us denote
\begin{center}
$\ds  \rho_t := \mathcal{T}_t^*(\rho_\Lambda) $
\end{center}
for every $t \geq 0$ (when the omission of the subindex does not cause any confusion). With this notation, equation (\ref{eq:gen-semigroup}) can be rewritten as the quantum dynamical master equation: 
\begin{center}
$\ds \partial_t \rho_t = \mathcal{L}_\Lambda^* (\rho_t)$. 
\end{center}

We say that a certain state $\sigma_\Lambda$ is an \textit{invariant state} of $  \qty{\mathcal{T}^*_t}_{t\geq 0} $  if 
\begin{center}
$\ds  \sigma_t :=  \mathcal{T}_t^*(\sigma_\Lambda)  = \sigma_\Lambda $
\end{center}
for every $t \geq 0$.

Throughout all this section, we will restrict to the \textit{primitive} case, i.e., $  \qty{\mathcal{T}^*_t}_{t\geq 0} $ has a unique full-rank invariant state. An interesting problem concerning quantum Markov semigroups is the study of the convergence to this unique invariant state, which can be done bounding the mixing time. 

The \textit{mixing time} of a quantum Markov semigroup is defined, given an initial state, as the time that the process spends to get close to the invariant state, i.e., the fixed point of the evolution. More specifically, it is given by the following expression
\begin{equation}
\tau (\epsilon)= \text{min} \qty{t>0 \, : \, \un{\rho_\Lambda \in \SSS_\Lambda}{\text{sup}} \,  \norm{\rho_t - \sigma_\Lambda}_1 \leq \epsilon  }.
\end{equation}

Let us assume that the quantum Markov proccess studied is \textit{reversible}, i.e., satisfies the \textit{detailed balance condition}:

\begin{center}
$\ds  \left\langle  f, \mathcal{L}_\Lambda(g) \right\rangle_{\sigma_\Lambda} = \left\langle  \mathcal{L}_\Lambda(f), g \right\rangle_{\sigma _\Lambda} $
\end{center}
for every $f, g \in \A_\Lambda$, where $\mathcal{L}_\Lambda$ is the generator of the evolution semigroup in the Heisenberg picture. 

Different bounds for the mixing time can be obtained by means of the optimal constants for some quantum functional inequalities, such as the \textit{spectral gap} for the \textit{Poincaré inequality} \cite{chi2} and the  \textit{logarithmic Sobolev constant} for the  \textit{logarithmic Sobolev inequality} \cite{kast-temme}.  In this section we will focus on the latter. There exists a whole family of \textit{logarithmic Sobolev inequalities} (\textit{log-Sobolev inequalities} for short), which can be indexed by an integer parameter, as done in \cite{kast-temme}. We will see below the well-known fact that the 1-log-Sobolev inequality, also known as modified log-Sobolev inequality, provides an upper bound for the mixing time. Since this is the only log-Sobolev inequality that will appear in this paper, we will just call it \textit{log-Sobolev inequality}, in a slight abuse of notation. 

The idea of bounding the mixing time in terms of log-Sobolev constants is based on two facts:
\begin{enumerate}
\item Finding a positive functional that bounds the convergence of the semigroup to the fixed point and bounding its derivative in terms of the functional itself. The role of this functional will be played by the relative entropy of $\rho_t$ and $\sigma_\Lambda$:
\begin{center}
$\ds  D(\rho_t || \sigma_\Lambda)= \tr[\rho_t (\log \rho_t - \log \sigma_\Lambda)]$.
\end{center}
\item Pinsker's inequality \cite{pinsker}: 
\begin{center}
$\ds   \norm{\rho_t - \sigma_\Lambda}_1 \leq  \sqrt{2 \ent {\rho_t} {\sigma_\Lambda}} $.
\end{center}
\end{enumerate}

Let us elaborate this first point. Since $\rho_t$ evolves according to $\LL_\Lambda^*$, the derivative of $D(\rho_t || \sigma_\Lambda)$ is given by
\begin{center}
$\partial_t D(\rho_t || \sigma_\Lambda) = \tr[\mathcal{L}_\Lambda^*(\rho_t) (\log \rho_t - \log \sigma_\Lambda)],$
\end{center}
which is a negative quantity (since the relative entropy of $\rho_t$ and $\sigma_\Lambda$ decreases with $t$), and we want to find a lower bound for the negative derivative of $D(\rho_t || \sigma_\Lambda)$ in terms of itself:
\begin{equation}\label{eq:bound-rel-ent}
2 \alpha D(\rho_t || \sigma_\Lambda) \leq -\tr[\mathcal{L}_\Lambda^*(\rho_t) (\log \rho_t - \log \sigma_\Lambda)].
\end{equation}

It is clear that, for each $\rho_t$, there exists an $\alpha$ that makes possible the previous inequality. However, finding a global $\alpha$ that works for every $\rho_t$ is far from trivial. Indeed, in general such quantity does not exist. A global constant for the previous inequality is called a \textit{log-Sobolev constant}.

\begin{defi}
Let $\LL_\Lambda^*$ be a Liouvillian in the Schrödinger picture and let $\sigma_\Lambda$ be the unique invariant full-rank state of the semigroup generated by $\LL_\Lambda^*$. We define the \textit{log-Sobolev constant} of $\mathcal{L}_\Lambda^*$ by
\begin{center}
$ \displaystyle \alpha(\mathcal{L}_\Lambda^*):= \underset{\rho_\Lambda \in \SSS_\Lambda}{\text{inf}}\frac{-\tr[\LL_\Lambda^*(\rho_\Lambda)(\log\rho_\Lambda-\log \sigma_\Lambda)]}{2 D(\rho_\Lambda || \sigma_\Lambda)}  $
\end{center}

\end{defi}

Assume that for a certain Liouvillian $\mathcal{L}^*_\Lambda$ a positive log-Sobolev constant exists. Then, we can integrate equation (\ref{eq:bound-rel-ent}) to obtain
\begin{equation}
D(\rho_t || \sigma_\Lambda)  \leq D(\rho_\Lambda || \sigma_\Lambda) e^{-2 \, \alpha(\mathcal{L}_\Lambda^*) \, t},
\end{equation}
and putting this together with Pinsker's inequality, we have:
\begin{equation}
\norm{\rho_t - \sigma_\Lambda}_1 \leq \sqrt{  2 D(\rho_\Lambda || \sigma_\Lambda) } \,  e^{-  \alpha(\mathcal{L}_\Lambda^*) \, t}.
\end{equation}

Finally, since $ D(\rho_\Lambda || \sigma_\Lambda) $ becomes maximal for a full-rank state $\sigma_\Lambda$, which is the case, when $\rho_\Lambda$ corresponds to a rank-one projector onto the minimal eigenvalue of $\sigma_\Lambda$, we obtain:
\begin{equation}
\norm{\rho_t - \sigma_\Lambda}_1 \leq \sqrt{  2 \log(1/ \sigma_\text{min})  } \,  e^{-  \alpha(\mathcal{L}_\Lambda^*) \, t}.
\end{equation}

Therefore, positive log-Sobolev constants can be used in upper bounds for the mixing time, providing an exponential improvement with respect to a bound in terms of the spectral gap. Proving whether a Lindbladian has a positive log-Sobolev constant is, thus, a fundamental problem in open quantum many-body systems.

In the following subsection, we will show that the heat-bath dynamics, with product fixed point, has a positive log-Sobolev constant. The global Lindbladian will be defined as the sum of local ones in the following form:
\begin{equation}\label{eq:lindblad-local}
\ds  \mathcal{L}^*_\Lambda = \un{x \in \Lambda}{\sum} \mathcal{T}_x^* - \identity_\Lambda ,
\end{equation}
where $ \mathcal{T}_x^*$ are certain quantum channels with a fixed point $\sigma_\Lambda$ verifying
\begin{equation}\label{eq:split-property}
\ds   \sigma_\Lambda = \un{x \in \Lambda}{\bigotimes} \, \sigma_x. 
\end{equation}

Our example constitutes a generalization of a particular case studied in \cite{depol-channels} and \cite{doubly-stoch}, where the authors consider doubly stochastic channels, i.e., channels verifying
\begin{center}
$\ds \mathcal{T}_x^* (\identity_\Lambda ) = \mathcal{T}_x (\identity_\Lambda ) = \identity_\Lambda  $,
\end{center}
and prove that, if the fixed point is $\sigma_\Lambda= \identity_\Lambda/ \text{dim}(\Lambda)$, the log-Sobolev constant of a Lindbladian of the form (\ref{eq:lindblad-local}) is lower bounded by $1/2$ and, hence, positive. Clearly, the identity verifies property (\ref{eq:split-property}), giving our result more generality in what concerns the fixed point. However, we only manage to prove positivity of the log-Sobolev constant for a certain channel (the Petz recovery map for the partial trace), whereas they obtain it for every channel with the identity as fixed point. 

Proving the existence of a positive log-Sobolev constant for a Lindbladian of the form (\ref{eq:lindblad-local}) for any quantum channel with fixed point satisfying  (\ref{eq:split-property}) is left as an open question.

\subsection{Example of positive log-Sobolev constant}

In this subsection, we show that the heat-bath dynamics, with product fixed point, has a positive log-Sobolev constant\footnote{After the completion of this work, we came to know that this constant was already proven to be positive in \cite[Theorem 9]{hypercontractivity}. However, in that result, the authors presented a lower bound  for the log-Sobolev constant that depends on some local gaps and the minimum eigenvalues of some local stationary states, whereas the bound that we give in this paper is universal and independent of any other quantity (indeed, it is $1/2$). Moreover, our proof is completely different and the techniques that we use in our paper are arguably simpler and allow us to establish an estrategy, based on quasi-factorization results for the relative entropy, that might be of use to prove positivity of log-Sobolev constants for more general dynamics in many-body systems.

Furthermore, simultaneously to our article, in an independent paper in the context of quantum hypothesis testing \cite{beigidattarouze}, the same theorem has also been obtained.}. Namely, given $\Lambda \subset \Z^d$ a quantum spin lattice,  if we take a product state 
\begin{center}
$\ds \sigma_\Lambda= \un{x \in \Lambda}{\bigotimes} \, \sigma_x $
\end{center}
on it, define for every $x \in \Lambda$ the minimal conditional expectation with respect to $\sigma_\Lambda$, $\E_x^*$,  as in Subsection \ref{subsec-cond-exp}, and consider the Lindbladian $\LL^*_x := \E_x^*- \identity_\Lambda $, then the global Lindbladian 
\begin{center}
$\ds \LL^*_\Lambda= \un{x \in \Lambda}{\sum} \LL^*_x $
\end{center}
is shown to have a positive log-Sobolev constant.

Let us first recall the definition of the minimal conditional expectation with respect to $\sigma_\Lambda$:
\begin{center}
$\ds \E^*_x (\rho_\Lambda)= \sigma_{\Lambda}^{1/2} \sigma_{x^c}^{-1/2} \rho_{x^c} \sigma_{x^c}^{-1/2} \sigma_{\Lambda}^{1/2}  $
\end{center} 
for all $\rho_\Lambda \in \SSS_{\Lambda}$. Since  $\sigma_\Lambda$ is a product state, we can write $\E_x^*(\rho_\Lambda)$ as

\begin{center}
$\ds \E^*_x (\rho_\Lambda)= \sigma_x  \otimes \rho_{x^c} $.
\end{center} 

Hence, for every $\rho_\Lambda \in \SSS_\Lambda$,
\begin{center}
$\ds  \LL^*_\Lambda ( \rho_\Lambda)= \un{x \in \Lambda}{\sum} \left(  \sigma_x  \otimes \rho_{x^c}  - \rho_\Lambda \right) $.
\end{center}

Noticing the definition of the global Lindbladian as the sum of local ones, one could think on the possibility of reducing the study of a quantity defined on the global Lindbladian to an analogous quantity defined on the Lindbladian associated to every site. Following this idea, we can define a \textit{conditional log-Sobolev constant}, on every subset $A \subset \Lambda$, as an auxiliary quantity for the proof of positivity of the global log-Sobolev constant.

\begin{defi}
Let $\Lambda \subset \Z^d$ be a finite lattice and let $\ds \LL^*_\Lambda=  \un{x \in \Lambda}{\sum} \LL^*_x$ be a global Lindbladian for the Schrödinger picture.  Given $A \subset \Lambda$, we define the \textit{conditional log-Sobolev constant} of $\LL^*_\Lambda$ in $A$ by 
\begin{center}
$ \displaystyle \alpha_\Lambda(\LL^*_A):= \underset{\rho_\Lambda \in \SSS_\Lambda}{\text{inf}}\frac{-\tr[\LL_A^*(\rho_\Lambda)(\log\rho_\Lambda-\log \sigma_\Lambda)]}{2 D_A(\rho_\Lambda || \sigma_\Lambda)}  $,
\end{center}
where $\sigma_\Lambda$ is the fixed point of the evolution, and $\ds  D_A(\rho_\Lambda || \sigma_\Lambda) $ is the conditional relative entropy.

\end{defi}

\begin{remark}
In Subsection \ref{subsec:comparison}, we have shown that, when $\sigma_\Lambda$ is product, both the conditional relative entropy and the conditional relative entropy by expectations coincide. Since that is the case studied in this subsection, any of them might be the one that appears in the definition of conditional log-Sobolev constant.

Indeed, for every $\rho_\Lambda \in \SSS_\Lambda$ and $A \subset \Lambda$,
\begin{center}
$\ds  \entA A {\rho_\Lambda} {\sigma_\Lambda} =  \enteA A {\rho_\Lambda} {\sigma_\Lambda}  = I_\rho (A : A^c) + \ent {\rho_A} {\sigma_A} $.
\end{center}

In this case, these definitions also coincide with the one that appears in \cite{bardet} and \cite{bardet-rouze}  under the name of \textit{decoherence-free relative entropy}.

\end{remark}

Now, we can prove the existence of a positive conditional log-Sobolev constant for every local Liouvillian in $x \in \Lambda$, $\LL^*_x$, and use this result to obtain a positive global log-Sobolev constant for $\LL^*_\Lambda$.

Taking a look at the definition of conditional log-Sobolev constant in $x\in \Lambda$, one can notice that the numerator of the global log-Sobolev constant comes from the sum of the conditional ones. However, the denominators lack a relation of this kind. Therefore, we need the following result of factorization of the relative entropy, which was proven in Subsection \ref{subsec:QF-sigmaprod}, to compare both conditional and global log-Sobolev constants:

\begin{theorem}
Let $\Lambda \subset \Z^d$ be a finite lattice and let $\rho_{\Lambda}, \sigma_{\Lambda} \in \SSS_\Lambda$ such that $\ds \sigma_\Lambda = \un{x \in \Lambda}{\bigotimes} \, \sigma_x$. The following inequality holds:
\begin{equation}
 \ent {\rho_{\Lambda}} {\sigma_{\Lambda}} \leq \un{x \in \Lambda}{\sum} \entA x {\rho_{\Lambda}} {\sigma_{\Lambda}}.
\end{equation}

\end{theorem}

In the following lemma we will prove that the Lindbladian defined at the beginning of this subsection has a positive conditional log-Sobolev constant. Indeed, we will show that this constant can be lower bounded by $1/2$. This, together with the previous result of factorization of the relative entropy, will be later used to prove positivity of the global log-Sobolev constant. 

\begin{lemma}\label{lemma:logSoblocal}
For every $x \in \Lambda$ and for $\LL^*_x$ defined as above, the following holds:
 \begin{center}
 $\ds \alpha_\Lambda(\LL^*_x) \geq \frac{1}{2}$.
 \end{center}
\end{lemma}

\begin{proof}
Let us write explicitly each term in the definition of $\alpha_\Lambda(\LL^*_x)$:
\begin{align*}
\ds \alpha_\Lambda(\LL^*_x) &=  \underset{\rho_\Lambda \in \SSS_\Lambda}{\text{inf}}\frac{-\tr[\LL_x^*(\rho_\Lambda)(\log\rho_\Lambda-\log \sigma_\Lambda)]}{2 D_x(\rho_\Lambda || \sigma_\Lambda)} \\
&= \underset{\rho_\Lambda \in \SSS_\Lambda}{\text{inf}}\frac{\tr[(\rho_\Lambda - \sigma_x \otimes \rho_{x^c} )(\log\rho_\Lambda-\log \sigma_\Lambda)]}{2 ( \ent {\rho_\Lambda} {\sigma_\Lambda} - \ent {\rho_{x^c}}  {\sigma_{x^c}})}\\
&=  \underset{\rho_\Lambda \in \SSS_\Lambda}{\text{inf}}\frac{ \ent {\rho_\Lambda} {\sigma_\Lambda} - \tr[ \sigma_x \otimes \rho_{x^c} \, (\log\rho_\Lambda-\log \sigma_\Lambda)]}{2 ( \ent {\rho_\Lambda} {\sigma_\Lambda} - \ent {\rho_{x^c}}  {\sigma_{x^c}})}.
\end{align*}

Consider now the second term in the numerator. Since $\sigma_\Lambda$, in particular, splits as a tensor product between the regions $x$ and $x^c$, we have:
\begin{align*}
&\tr[ \sigma_x \otimes \rho_{x^c} \, (\log\rho_\Lambda-\log \sigma_\Lambda)] =\\
&\phantom{sadadasdad}= \tr[ \sigma_x \otimes \rho_{x^c} \, (\log\rho_\Lambda - \log \sigma_x \otimes \rho_{x^c}+\log \sigma_x \otimes \rho_{x^c}-\log \sigma_x \otimes \sigma_{x^c})]\\
&\phantom{sadadasdad}= \tr[ \sigma_x \otimes \rho_{x^c} \, (\log\rho_\Lambda - \log \sigma_x \otimes \rho_{x^c})] + \tr[  \rho_{x^c} \, (\log \rho_{x^c}-\log \sigma_{x^c})]\\
&\phantom{sadadasdad}= - \ent { \sigma_x \otimes \rho_{x^c}} {\rho_\Lambda} + \ent {\rho_{x^c}} {\sigma_{x^c}}.
\end{align*}

Therefore, $\alpha_\Lambda(\LL_x)$ is given by:

\begin{align*}
\ds \alpha_\Lambda(\LL^*_x) &=   \underset{\rho_\Lambda \in \SSS_\Lambda}{\text{inf}}\frac{ \ent {\rho_\Lambda} {\sigma_\Lambda} + \ent { \sigma_x \otimes \rho_{x^c}} {\rho_\Lambda} - \ent {\rho_{x^c}} {\sigma_{x^c}}}{2 ( \ent {\rho_\Lambda} {\sigma_\Lambda} - \ent {\rho_{x^c}}  {\sigma_{x^c}})} \\
&= \frac{1}{2} +  \underset{\rho_\Lambda \in \SSS_\Lambda}{\text{inf}}\frac{  \ent { \sigma_x \otimes \rho_{x^c}} {\rho_\Lambda}}{2 ( \ent {\rho_\Lambda} {\sigma_\Lambda} - \ent {\rho_{x^c}}  {\sigma_{x^c}})} \\
& \geq \frac{1}{2},
\end{align*}
since $ \ent {\rho_\Lambda} {\sigma_\Lambda} - \ent {\rho_{x^c}}  {\sigma_{x^c}} \geq 0$ (Property of monotonicity of the relative entropy) and $ \ent { \sigma_x \otimes \rho_{x^c}} {\rho_\Lambda} \geq 0$ (Property of non-negativity of the relative entropy).

\end{proof}

Finally, we are in position of proving positivity of the global log-Sobolev constant from the previous lemma and Theorem \ref{thm:factorization}.

\begin{theorem}
 $\LL_\Lambda^*$ defined as above has a global positive log-Sobolev constant.

\end{theorem}

\begin{proof}
In virtue of the result of factorization proven above (Theorem \ref{thm:factorization}), we know that
\begin{equation}\label{eq:fact2}
 \ent {\rho_{\Lambda}} {\sigma_{\Lambda}} \leq \un{x \in \Lambda}{\sum} \entA x {\rho_{\Lambda}} {\sigma_{\Lambda}}
\end{equation}
for every $\rho_\Lambda \in \SSS_\Lambda$.

From the definition of $\alpha_\Lambda(\LL^*_x)$, it is clear that the following holds for every $x \in \Lambda$

\begin{center}
$\ds  \entA x {\rho_{\Lambda}} {\sigma_{\Lambda}} \leq \frac{-\tr[\LL_x^*(\rho_\Lambda)(\log \rho_\Lambda - \log \sigma_\Lambda)]}{\alpha_\Lambda(\LL_x)} $.
\end{center}

Putting this together with equation (\ref{eq:fact2}), we have:
\begin{align*}
\ent {\rho_{\Lambda}} {\sigma_{\Lambda}}& \leq \un{x \in \Lambda}{\sum} \entA x {\rho_{\Lambda}} {\sigma_{\Lambda}} \\
& \leq \un{x \in \Lambda}{\sum} \frac{-\tr[\LL_x^*(\rho_\Lambda)(\log \rho_\Lambda - \log \sigma_\Lambda)]}{\alpha_\Lambda(\LL^*_x)} \\
& \leq \frac{1}{\un{x \in \Lambda}{\text{inf}} \, \alpha_\Lambda(\LL^*_x)} \un{x \in \Lambda}{\sum} -\tr[\LL_x^*(\rho_\Lambda)(\log \rho_\Lambda - \log \sigma_\Lambda)] \\
&= \frac{1}{\un{x \in \Lambda}{\text{inf}} \, \alpha_\Lambda(\LL^*_x)} \left( - \tr[ \LL_\Lambda^*(\rho_\Lambda)(\log \rho_\Lambda - \log \sigma_\Lambda) ] \right) \\
& \leq 2 \left( - \tr[ \LL_\Lambda^*(\rho_\Lambda)(\log \rho_\Lambda - \log \sigma_\Lambda) ] \right), 
\end{align*}
where, in the fourth line, we have used the definition of $\LL_\Lambda^*$ and, in the fifth line, Lemma \ref{lemma:logSoblocal}. This expression holds for every $\rho_\Lambda \in \SSS_\Lambda$.

Finally, recalling the definition of $\alpha(\LL^*_\Lambda)$, we have

\begin{center}
$\ds \alpha(\LL^*_\Lambda) = \underset{\rho_\Lambda \in \SSS_\Lambda}{\text{inf}}\frac{-\tr[\LL_\Lambda^*(\rho_\Lambda)(\log\rho_\Lambda-\log \sigma_\Lambda)]}{2 D(\rho_\Lambda || \sigma_\Lambda)} \geq \frac{1}{2}$.
\end{center}

Hence,  $\LL^*_\Lambda$  has a global positive log-Sobolev constant, which is greater or equal than $1/2$.
\end{proof}

\begin{remark}
The structure of the proof followed to obtain positivity for the log-Sobolev constant is an analogous quantum version of a simplification to the one used in \cite{clasico} and \cite{cesi} to prove a bound on a log-Sobolev constant that connects  the decay of correlations in the Gibbs state of a classical spin model to the mixing time of the associated Glauber dynamics. One could then hope that the results of quasi-factorization of the relative entropy of the previous sections might be of use to obtain positive log-Sobolev constants for certain dynamics and connect it with a decay of correlations on the Gibbs state above the critical temperature. This is left for future work.

\end{remark}

\section{Conclusions}

In this paper, we have introduced a new quantity in quantum information theory, the \textit{conditional relative entropy}, and we have characterized it axiomatically, as well as presented several results of quasi-factorization of the relative entropy in terms of this conditional relative entropy.

Afterwards, we have weakened the definition of conditional relative entropy and presented an example of this weaker definition, which we have called \textit{conditional relative entropy by expectations}. We have compared both definitions, seen that both extend their classical analogue, and presented for the latter some weaker results of quasi-factorization.

Finally, in the last part of the paper, we have seen that a result of quasi-factorization is the key tool to prove that the heat-bath dynamics, with product fixed point, has a positive log-Sobolev constant. 

However, throughout the whole manuscript, we have also left several open questions. We proceed now to discuss them in more detail.

The main open question that yields from this work is related to the main result of Section \ref{sec:logSob}. More specifically, in that section, we show that a result of quasi-factorization of the relative entropy, when the second state is a tensor product, is the key tool to prove that the heat-bath dynamics, with product fixed point, has a positive log-Sobolev constant. Whether the same strategy might be followed to obtain a positive log-Sobolev constant for the heat-bath dynamics, in a more general setting, from the stronger result of quasi-factorization presented in Subsection \ref{subsec:QF-CRE}, under the assumption of a decay of correlations in the fixed point, is left as an open question.

\begin{problem}
Use the result of quasi-factorization of the relative entropy in terms of conditional relative entropies of Subsection \ref{subsec:QF-CRE} to obtain positive log-Sobolev constant for the heat-bath dynamics, with a general fixed point $\sigma$.
\end{problem}

This result of quasi-factorization could also be adapted to the Davies generator setting, and the same strategy might lead to a positive log-Sobolev constant for this dynamics. 

\begin{problem}
Can one adapt the result of Subsections \ref{subsec:QF-CRE} or \ref{subsec:QF-CREexp} to the Davies generator, to obtain a positive log-Sobolev constant for this dynamics? 
\end{problem}

To tackle this problem, it is likely that we first need to improve the result of quasi-factorization for the conditional expectation.

\begin{problem}
Improve the result of quasi-factorization of the conditional relative entropy by expectations of Subsection \ref{subsec:QF-CREexp}, by improving the bound that we obtained for the error term. 
\end{problem}

Concerning the definition of conditional relative entropy presented in this paper, we have shown several clues that allow us to think that the definition is reasonable. However, there is some space to possibly improve it, in some sense, so that we can obtain results of quasi-factorization more easily, for example.

\begin{problem}
Improve the definition of conditional relative entropy. One idea to do that could be to add the property proven in equation (\ref{eq:prop-sigma-prod}) to the definition.

Is there any possible axiomatic definition for conditional relative entropy from which we can immediately obtain results of quasi-factorization?
\end{problem}

Moreover, in Subsection \ref{subsec:comparison}, we have compared the definitions of conditional relative entropy and conditional relative entropy by expectations. On the one side, we have shown several cases where they coincide, and on the other side, we have seen that this cannot hold always. We leave the possibility of studying in general for which cases both expressions are the same as an open problem:

\begin{problem}
Characterize for which states $\rho_{AB}, \sigma_{AB} \in \SSS_{AB}$, the following holds:
\begin{center}
$\ds  \entA A {\rho_{AB}} {\sigma_{AB}} = \enteA A  {\rho_{AB}} {\sigma_{AB}} $,
\end{center}
or, at least, find more examples where this equality holds. 
\end{problem}

Finally, when introducing the result of positivity of the log-Sobolev constant for the heat-bath dynamics with product fixed point, we have mentioned that proving the existence of a positive log-Sobolev constant for a more general Lindbladian of the same form (sum of local terms) for any quantum channel with product fixed point is left as an open question.

\begin{problem}
For $\hs_\Lambda= \un{x \in \Lambda}{\bigotimes} \hs_x $, and $\sigma_\Lambda=\un{x \in \Lambda}{\bigotimes} \sigma_x $, prove, if true, that, if $\mathcal{T}_x^*$ is a quantum channel with $\sigma_x$ as fixed point for every $x \in \Lambda$, then
\begin{center}
$\ds  \un{x \in \Lambda}{\sum} \mathcal{T}_x^* - \identity_\Lambda $
\end{center}
has a positive log-Sobolev constant.
\end{problem}

\section*{Acknowledgment}

The authors would like to thank David Sutter, Ivan Bardet and Andreas Bluhm for very helpful conversations. AC and DPG acknowledge support from MINECO (grants MTM2014-54240-P and MTM2017-88385-P), from Comunidad de Madrid (grant QUITEMAD+- CM, ref. S2013/ICE-2801), and the European Research Council (ERC) under the European Union’s Horizon 2020 research and innovation programme (grant agreement No 648913). AC is partially supported by a La Caixa-Severo Ochoa grant (ICMAT Severo Ochoa project SEV-2011-0087, MINECO). AL acknowledges financial support from the European Research Council (ERC Grant Agreement no 337603), the Danish Council for Independent Research (Sapere Aude) and VILLUM FONDEN via the QMATH Centre of Excellence (Grant No. 10059). This work has been partially supported by ICMAT Severo Ochoa project SEV-2015-0554 (MINECO).

\end{document}